\documentclass[11pt,letterpaper]{article}

\usepackage[hmargin=2.3cm,vmargin=2.5cm]{geometry}
\usepackage{graphicx}
\usepackage{amsmath,amsthm,amssymb}
\usepackage{comment}
\usepackage[ruled]{algorithm}
\usepackage[]{algorithmic}
\usepackage[caption=false,font=footnotesize]{subfig}
\usepackage{multirow}
\usepackage{url}
\usepackage{color}
\usepackage{paralist}
\usepackage{tabularx}
\usepackage{booktabs}
\usepackage{epstopdf}
\usepackage{authblk}

\newtheorem{mydef}{Definition}
\newtheorem{mypro}{Proposition}
\newtheorem{mythe}{Theorem}
\DeclareMathOperator{\cut}{\texttt{cut}}
\DeclareMathOperator{\links}{\texttt{links}}
\DeclareMathOperator{\cond}{\texttt{cond}}
\DeclareMathOperator{\ncut}{\texttt{ncut}}
\let\deg=\relax
\DeclareMathOperator{\deg}{\texttt{deg}}
\DeclareMathOperator{\dist}{\texttt{dist}}

\usepackage{xspace}
\newcommand{\NISE}{\textsc{nise}\xspace}

\def\clap#1{\hbox to 0pt{\hss#1\hss}}
\def\mathclap{\mathpalette\mathclapinternal}
\def\mathclapinternal#1#2{%
\clap{$\mathsurround=0pt#1{#2}$}}

\title{Overlapping Community Detection Using \\ Neighborhood-Inflated Seed Expansion}

\vspace{0.3cm}

\author[1]{Joyce Jiyoung Whang}
\author[2]{David F. Gleich}
\author[1]{Inderjit S. Dhillon}
\affil[1]{Dept. of Computer Science\\
University of Texas at Austin\\
\{joyce,inderjit\}@cs.utexas.edu}
\affil[2]{Dept. of Computer Science\\
Purdue University\\
dgleich@purdue.edu}

\providecommand{\keywords}[1]{\textbf{\textit{Index Terms---}} #1}

\date{}

\begin{document}

\maketitle

\vspace{-1cm}

\begin{abstract}
Community detection is an important task in network analysis. A community (also referred to as a cluster) is a set of cohesive vertices that have more connections inside the set than outside. In many social and information networks, these communities naturally overlap. For instance, in a social network, each vertex in a graph corresponds to an individual who usually participates in multiple communities. In this paper, we propose an efficient overlapping community detection algorithm using a seed expansion approach. The key idea of our algorithm is to find good seeds, and then greedily expand these seeds based on a community metric. Within this seed expansion method, we investigate the problem of how to determine good seed nodes in a graph. In particular, we develop new seeding strategies for a personalized PageRank clustering scheme that optimizes the conductance community score. Experimental results show that our seed expansion algorithm outperforms other state-of-the-art overlapping community detection methods in terms of producing cohesive clusters and identifying ground-truth communities. We also show that our new seeding strategies are better than existing strategies, and are thus effective in finding good overlapping communities in real-world networks.

\end{abstract}


\section{Introduction}
Community detection is one of the most important and fundamental tasks in network analysis with applications in functional prediction in biology~\cite{Lee-2013-community} and sub-market identification~\cite{andersen-2006-seed} among others. Given a network, a community is defined to be a set of cohesive nodes that have more connections inside the set than outside. Since a network can be modelled as a graph with vertices and edges, community detection can be thought as a graph clustering problem where each community corresponds to a cluster in the graph. In this manuscript, the terms \textit{cluster} and \textit{community} are used interchangeably.

The goal of traditional, exhaustive graph clustering algorithms (e.g., Metis \cite{karypis-1998-metis}, Graclus \cite{dhillon-2007-graclus}) is to partition a graph such that every node belongs to exactly one cluster. However, in many social and information networks, nodes participate in multiple communities. For instance, in a social network, nodes represent individuals and edges represent social interactions between the individuals. In this setting, a node's communities can be interpreted as its social circles. Thus, it is likely that a node belongs to multiple communities, i.e., communities naturally overlap. To find these groups, we study the problem of overlapping community detection where communities are allowed to overlap with each other and some nodes are allowed not to belong to any cluster.

In this paper, we propose an efficient overlapping community detection algorithm using a seed expansion approach. More specifically, we investigate how to select good seeds in a method that grows communities around seeds. These local expansion methods are among the most successful strategies for overlapping community detection \cite{xie-2012-survey}. However, principled methods to choose the seeds are few and far between. When they exist, they are usually computationally expensive, for instance, using maximal cliques as seeds \cite{shen-2009-eagle}. Empirically successful strategies include exhaustively exploring all individual seeds and greedy methods that randomly pick a vertex, grow a cluster, and continue with any unassigned vertex.

To find a set of good seeds, we present two effective seeding strategies which we call ``Graclus centers'' and ``Spread hubs.'' The ``Graclus centers'' seeding is based on the same distance kernel that underlies the equivalence between kernel $k$-means and graph clustering objectives \cite{dhillon-2007-graclus}. Using this distance function, we can efficiently locate a good seed \emph{within} an existing set of cohesive vertices of the graph. Specifically, we first compute many clusters using a multi-level weighted kernel $k$-means algorithm on the graph (the Graclus algorithm) \cite{dhillon-2007-graclus}, then use the corresponding distance function to compute the ``centroid vertex'' of each cluster. We use the neighborhood set of each centroid vertex as a seed region for community detection. The idea of ``Spread hubs'' seeding is to select an independent set of high degree vertices. This seeding strategy is inspired by the recent observations that there should be good clusters around high degree vertices in many real-world networks which have a power-law degree distribution \cite{whang-2012-gem}, \cite{gleich-2012-neighborhoods}. 


The algorithm we use to grow a seed set is based on personalized PageRank (PPR) clustering \cite{andersen-2006-local}. The high level idea of this expansion method is to first compute the PPR vector for each of the seeds, and then expand each seed based on the PPR score. It is important to note that we can have multiple nodes in the personalization vector, and indeed we use the entire vertex neighborhood of a seed node as the personalization vector for PPR. This \textit{neighborhood inflation} plays a critical role in the success of our algorithm. The full algorithm to compute overlapping clusters from the seeds is discussed in Section \ref{sec:algo}. We name our algorithm \NISE by abbreviating our main idea, Neighborhood-Inflated Seed Expansion.


Our experimental results show that our seeding strategies are better than existing seeding strategies, and effective in finding good overlapping communities in real-world networks. More importantly, we observe that \NISE significantly outperforms other state-of-the-art overlapping community detection methods in terms of producing cohesive clusters and identifying ground-truth communities. Also, our method scales to problems with over 45 million edges, whereas other existing methods were unable to complete on these large datasets.

\section{Preliminaries} 
\label{sec:pre}
In this section, we formally describe the overlapping community detection problem, and review some important concepts in graph clustering. Also, we introduce real-world networks which are used in our experiments. 

\subsection{Problem Statement}
Given a graph $G=(\mathcal{V}, \mathcal{E})$ with a vertex set $\mathcal{V}$ and an edge set $\mathcal{E}$, we can represent the graph as an adjacency matrix $A$ such that $A_{ij}=e_{ij}$ where $e_{ij}$ is the edge weight between vertices $i$ and $j$, or $A_{ij}=0$ if there is no edge. We assume that graphs are undirected, i.e., $A$ is symmetric. The goal of the traditional, exhaustive graph clustering problem is to partition a graph into $k$ pairwise disjoint clusters $\mathcal{C}_1, \cdots, \mathcal{C}_k$ such that $\mathcal{C}_1 \cup \cdots \cup \mathcal{C}_k = \mathcal{V}$. On the other hand, the goal of the overlapping community detection problem is to find overlapping clusters whose union is not necessarily equal to the entire vertex set $\mathcal{V}$. Formally, we seek $k$ overlapping clusters such that $\mathcal{C}_1 \cup \cdots \cup \mathcal{C}_k \subseteq \mathcal{V}$. 

\subsection{Measures of Cluster Quality}
\label{nonMetrics}

There are some popular measures for gauging the quality of clusters: cut, normalized cut, and conductance. Let us define $\links(\mathcal{C}_p, \mathcal{C}_q)$ to be the sum of edge weights between vertex sets $\mathcal{C}_p$ and $\mathcal{C}_q$.

\textbf{Cut.} The cut of cluster $\mathcal{C}_i$ is defined as the sum of edge weights between $\mathcal{C}_i$ and its complement, $\mathcal{V} \backslash \mathcal{C}_i$. That is, 
\begin{equation}
\label{cut}
\cut(\mathcal{C}_i) = \links(\mathcal{C}_i, \mathcal{V} \backslash \mathcal{C}_i). 
\end{equation}

\textbf{Normalized Cut.} The normalized cut of a cluster is defined by the cut with volume normalization as follows:
\begin{equation}
\label{ncut}
\ncut(\mathcal{C}_i) = \dfrac{\cut(\mathcal{C}_i)}{\links(\mathcal{C}_i, \mathcal{V})}.
\end{equation}

\textbf{Conductance.} The conductance of a cluster is defined to be the cut divided by the least number of edges incident on either set $\mathcal{C}_i$ or $\mathcal{V} \backslash \mathcal{C}_i$:
\begin{equation}
\nonumber
\label{cond_def}
\cond(\mathcal{C}_i) = \dfrac{\cut(\mathcal{C}_i)}{ \min \bigg(  \links(\mathcal{C}_i, \mathcal{V}), \links(\mathcal{V} \backslash \mathcal{C}_i, \mathcal{V}) \bigg) }.
\end{equation}
By definition, $\cond(\mathcal{C}_i)=\cond(\mathcal{V} \backslash \mathcal{C}_i)$. The conductance of a cluster is the probability of leaving that cluster by a one-hop walk starting from the smaller set between $\mathcal{C}_i$ and $\mathcal{V} \backslash \mathcal{C}_i$. Notice that $\cond(\mathcal{C}_i)$ is always greater than or equal to $\ncut(\mathcal{C}_i)$. 

\subsection{Graph Clustering and Weighted Kernel $k$-means}
\label{sec:wkkm}
It has been shown that a graph clustering objective is mathematically equivalent to a weighted kernel $k$-means objective \cite{dhillon-2007-graclus}. For example, let us consider the normalized cut objective of a graph $G$ which is defined to be
\begin{equation}
\label{gncut}
\ncut(G) = \min_{\mathcal{C}_1, ..., \mathcal{C}_k} \sum_{i=1}^{k} \dfrac{\links(\mathcal{C}_i, \mathcal{V} \backslash \mathcal{C}_i)}{\links(\mathcal{C}_i, \mathcal{V})}. 
\end{equation}
This objective can be shown to be equivalent to a weighted kernel $k$-means objective by defining a weight for each data point to be the degree of a vertex, and the kernel matrix to be $K=\sigma D^{-1} + D^{-1}AD^{-1}$, where $D$ is the diagonal matrix of degrees (i.e., $D_{ii}=\sum_{j=1}^n A_{ij}$ where $n$ is the total number of nodes), and $\sigma$ is a scalar typically chosen to make $K$ positive-definite. Then, we can quantify the kernel distance between a vertex $v \in \mathcal{C}_i$ and cluster $\mathcal{C}_i$, denoted $\dist(v, \mathcal{C}_i)$, as follows:\\
\begin{minipage}{\linewidth}
\small
\begin{align}
\label{dist}
& \dist(v, \mathcal{C}_i) = \\
& -\dfrac{2 \links(v, \mathcal{C}_i)}{\deg(v) \deg(\mathcal{C}_i)} + \dfrac{\links(\mathcal{C}_i, \mathcal{C}_i)}{\deg(\mathcal{C}_i)^2} + \dfrac{\sigma}{\deg(v)} - \dfrac{\sigma}{\deg(\mathcal{C}_i)} \nonumber 
\end{align}
\end{minipage}\\
where $\deg(v)=\links(v,\mathcal{V})$, and $\deg(\mathcal{C}_i)=\links(\mathcal{C}_i, \mathcal{V})$.

\begin{table*}[htbp]
{\scriptsize
    \caption{Summary of Real-world Networks.}
    \centering
    \begin{tabularx}{\linewidth}{ccccccccc}
    \toprule
    Category & Graph & No. of vertices & No. of edges & Max. Deg. & Avg. Deg. & Avg. CC & Ground-truth & Source \\
    \midrule
    Collaboration & HepPh & 11,204 & 117,619 & 491 & 21.0 & 0.6216 & N/A & \cite{snap}\\ 
    & AstroPh & 17,903 & 196,972 & 504 & 22.0 & 0.6328 & N/A & \cite{snap}\\ 
    & CondMat & 21,363 & 91,286 & 279 & 8.5 & 0.6417 & N/A & \cite{snap}\\ 
    & DBLP & 317,080 & 1,049,866 & 343 & 6.6 & 0.6324 & \checkmark & \cite{snap}\\ 
    \midrule
    Product & Amazon & 334,863 & 925,872 & 549 & 5.5 & 0.3967 & \checkmark & \cite{snap}\\
    \midrule
    Social & Orkut & 731,332 & 21,992,171 & 6,933 & 60.1 & 0.2468 & \checkmark & \cite{snap}\\
    & Flickr & 1,994,422 & 21,445,057 & 27,908 & 21.5 & 0.1881 & N/A & \cite{mislove-2008-flickr}\\ 
    & Myspace & 2,086,141 & 45,459,079 & 92,821 & 43.6 & 0.1242 & N/A & \cite{song-2012-mslj}\\ 
    & LiveJournal & 1,757,326 & 42,183,338 & 29,771 & 48.0 & 0.2400 & N/A & \cite{song-2012-mslj}\\ 
    & LiveJournal2 & 1,143,395 & 16,880,773 & 11,495 & 29.5 & 0.2535 & \checkmark & \cite{snap}\\    
    \bottomrule
    \end{tabularx}
 	\label{input}
}
\end{table*}

\begin{figure}[htbp]
\centering
\includegraphics[width=0.33\textwidth]{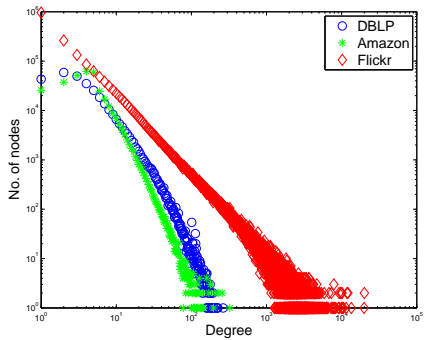}
\caption{Degree distributions of real-world networks -- the degree distributions follow a power-law.}
\label{degdist}
\end{figure}

\subsection{Datasets}
\label{sec:data}
We use ten different real-world networks including collaboration networks, social networks, and a product network from \cite{snap}, \cite{song-2012-mslj}, and \cite{mislove-2008-flickr}. The networks are presented in Table \ref{input}. All the networks are loop-less, connected, undirected graphs. 

In a collaboration network, vertices indicate authors, and edges indicate co-authorship. If authors $u$ and $v$ wrote a paper together, there exists an edge between them. For example, if a paper is written by three authors, this is represented by a clique of size three in the network. HepPh, AstroPh, and CondMat networks are constructed based on the papers submitted to arXiv e-print service. Specifically, HepPh represents the High Energy Physics (Phenomenology) category, AstroPh represents the Astrophysics category, and CondMat represents the Condensed Matter Physics category. The DBLP network is constructed based on the DBLP computer science bibliography website.

We use five different social networks: Flickr, Myspace, LiveJournal, LiveJournal2 (a variation with ground-truth), and Orkut. Flickr is an online photo sharing application, Myspace is a social entertainment networking service, LiveJournal is a blogging application where users can publish their own journals, and Orkut was a social networking website operated by Google. Users can make a friendship relationship with each other in each of these websites. So, in these social networks, nodes represent users and edges represent friendship relationships between them.

In the Amazon product network, vertices represent products and edges represent co-purchasing information. If products $u$ and $v$ are frequently co-purchased, there exists an undirected edge between them. This network is constructed based on \textit{Customers Who Bought This Item Also Bought} feature of the Amazon website.
 
In Table \ref{input}, we present the number of nodes/edges, the maximum degree, the average degree, and the average clustering coefficient (CC) of each of the networks. Figure \ref{degdist} shows the degree distributions of DBLP, Flicker and Amazon networks. We can see that the real-world networks have distinguishing characteristics: a power-law degree distribution \cite{barabasi-1999-powerlaw} and a high clustering coefficient \cite{watts-1998-largecc}, \cite{easley-2010-sna}. 

As indicated in Table \ref{input}, we have ground-truth communities \cite{snap} on some of the datasets. In DBLP, each publication venue (i.e., journal or conference) can be considered as an individual ground-truth community. In the Amazon network, each ground-truth community can be defined to be a product category that Amazon provides. In LiveJournal2 and Orkut networks, there exists user-defined social groups. On LiveJournal2 and Orkut networks, the ground-truth communities do not cover a substantial portion of the graph, so we use a subgraph which is induced by the nodes that have at least one membership in the ground-truth communities. In Table \ref{input}, the statistics about LiveJournal2 and Orkut are based on the induced subgraphs we used in our experiments.



\section{Overlapping Community Detection Using Neighborhood-Inflated Seed Expansion} 
\label{sec:algo}
We introduce our overlapping community detection algorithm, \NISE. It consists of four phases: filtering, seeding, seed expansion, and propagation. In the filtering phase, we remove regions of the graph that are trivially separable from the rest of the graph. In the seeding phase, we find good seeds in the filtered graph, and in seed expansion phase, we expand the seeds using a personalized PageRank clustering scheme. Finally, in the propagation phase, we further expand the communities to the regions that were removed in the filtering phase.


\subsection{Filtering Phase}
\label{sec:preprocessing}
The goal of the filtering phase is to identify regions of the graph where an algorithmic solution is required to identify the overlapping clusters. To explain our filtering step, recall that almost all graph partitioning methods begin by assigning each connected component to a separate partition. Any other choice of partitioning for disconnected components is entirely arbitrary. The Metis procedure \cite{karypis-1998-metis}, for instance, may combine two disconnected components into a single partition in order to satisfy a balance constraint on the partitioning. For the problem of overlapping clustering, an analogous concept can be derived from biconnected components. Formally, a biconnected component is defined as follows:

\begin{mydef}
Given a graph $G = (\mathcal{V}, \mathcal{E})$, a biconnected component is a maximal induced subgraph $G' = (\mathcal{V}', \mathcal{E}')$ that remains connected after removing any vertex and its adjacent edges in $G'$.
\end{mydef}

Let us define the size of a biconnected component to be the number of edges in $G'$. Now, consider all the biconnected components of size one. Notice that there should be no overlapping partitions that use these edges because they bridge disjoint communities. Consequently, our filtering procedure is to find the largest connected component of the graph after we remove all single-edge biconnected components. We call this the ``biconnected core" of the graph even though it may not be biconnected. Let $\mathcal{E}_S$ denote all the single-edge biconnected components. Then, the biconnected core graph is defined as follows:

\begin{mydef}
The biconnected core $G_C = (\mathcal{V}_C, \mathcal{E}_C)$ is the maximum size connected subgraph of $G'' = (\mathcal{V}, \mathcal{E} \setminus \mathcal{E}_S)$.
\end{mydef} 

Notice that the biconnected core is not the 2-core of the original graph (a $k$-core graph is a maximal subgraph of the original graph in which all nodes have degree at least $k$ \cite{seidman-1983-kcore}). Subgraphs connected to the biconnected core are called \emph{whiskers} by Leskovec et al.~\cite{leskovec-2009-natural} and we use the concept of a bridge to define them:
\begin{mydef}
A bridge is a biconnected component of size one which is directly connected to the biconnected core.
\end{mydef}
\noindent Whiskers are then defined as follows:
\begin{mydef}
A whisker $W = (\mathcal{V}_W, \mathcal{E}_W)$ is a maximal subgraph of $G$ that can be detached from the biconnected core by removing a bridge, 
\end{mydef} 

Let $\mathcal{E}_B$ be all the bridges in a graph. Notice that $\mathcal{E}_B \subseteq \mathcal{E}_S$. On the region which is not included in the biconnected core graph $G_C$, we define the detached graph $G_D$ as follows:

\begin{mydef}
$G_D = (\mathcal{V}_D, \mathcal{E}_D)$ is the subgraph of $G$ which is induced by $\mathcal{V} \setminus \mathcal{V}_C$.
\end{mydef}

Finally, given the original graph $G = (\mathcal{V}, \mathcal{E})$, $\mathcal{V}$ and $\mathcal{E}$ can be decomposed as follows:

\begin{mypro}
Given a graph $G = (\mathcal{V}, \mathcal{E})$, $\mathcal{V}=\mathcal{V}_C \cup \mathcal{V}_D$ and $\mathcal{E} = \mathcal{E}_C \cup \mathcal{E}_D \cup \mathcal{E}_B$.
\label{propo}
\end{mypro}

\begin{proof}
This follows from the definitions of the biconnected core, bridges, and the detached graph.
\end{proof}

\begin{figure}[t]
\centering
\includegraphics[width=0.35\textwidth]{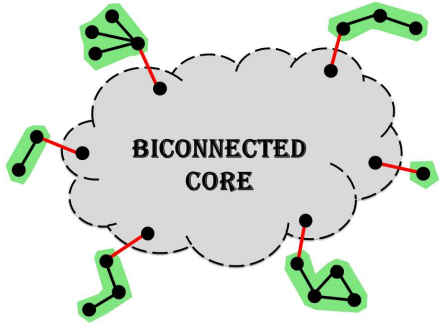}
\caption{Biconnected core, whiskers, and bridges -- grey region indicates the biconnected core where vertices are densely connected to each other, and green components indicate whiskers. Red edges indicate bridges which connect the biconnected core and each of the whiskers.}
\label{biconncore}
\end{figure}

\begin{table*}[t]
{\scriptsize
\caption{Biconnected core and the detached graph (in the last column, LCC refers to the largest connected component).}
\begin{tabularx}{\linewidth}{XXXXX}
\toprule
 &\multicolumn{2}{l}{Biconnected core} &\multicolumn{2}{l}{Detached graph}  \\
 \cmidrule(l){2-3} \cmidrule(l){4-5} 
 & No. of vertices (\%) & No. of edges (\%) & No. of components & Size of the LCC (\%) \\
 \midrule
 HepPh & 9,945 (88.8\%) & 116,099 (98.7\%) & 1,123 &  21 (0.1874\%) \\ 
 AstroPh & 16,829 (94.0\%) &  195,835 (99.4\%) & 957 &  23 (0.1285\%) \\ 
 CondMat & 19,378 (90.7\%) & 89,128 (97.6\%) & 1,669 &  12 (0.0562\%)  \\ 
 DBLP & 264,341 (83.4\%) & 991,125 (94.4\%) & 43,093 & 32 (0.0101\%)  \\ 
 Amazon &  291,449 (87.0\%) & 862,836 (93.2\%) & 25,835 &  250 (0.0747\%) \\
 Flickr & 954,672 (47.9\%) & 20,390,649 (95.1\%) & 864,628 & 107 (0.0054\%) \\
 Myspace & 1,724,184 (82.7\%) & 45,096,696 (99.2\%) & 332,596 & 32 (0.0015\%)  \\
 LiveJournal & 1,650,851 (93.9\%) & 42,071,541 (99.7\%) & 101,038 & 105 (0.0060\%) \\ 
 LiveJournal2 & 1,076,499 (94.2\%) & 16,786,580 (99.4\%) & 59,877 & 91 (0.0080\%) \\
 Orkut & 729,634 (99.8\%) & 21,990,221 (99.9\%) & 1,529 & 15 (0.0021\%) \\
 \bottomrule
\end{tabularx}
\label{bicore}
}
\end{table*}

Figure \ref{biconncore} illustrates the biconnected core, whiskers, and bridges. The output of our filtering phase is the biconnected core graph where whiskers are filtered out. The filtering phase removes regions of the graph that are clearly partitionable from the remainder. Note that there is no overlap between any of the whiskers. This indicates that there is no need to apply overlapping community detection algorithm on the detached regions.

Table \ref{bicore} shows the size of the biconnected core and the connectivity of the detached graph in our real-world networks. Details of these networks are presented in Table \ref{input}. We compute the size of the biconnected core in terms of the number of vertices and edges. The number reported in the parenthesis shows how many vertices or edges are included in the biconnected core, i.e., the percentages of $|\mathcal{V}_C|$/$|\mathcal{V}|$ and $|\mathcal{E}_C|$/$|\mathcal{E}|$, respectively. We also compute the number of connected components in the detached graph, and the size of the largest connected component (LCC in Table \ref{bicore}) in terms of the number of vertices. The number reported in the parenthesis indicates the relative size of the largest connected component compared to the number of vertices in the original graph.

We can see that the biconnected core contains a substantial portion of the edges. In terms of the vertices, the biconnected core contains around 80 or 90 percentage of the vertices for all datasets except Flickr. In Flickr, the biconnected core only contains around 50 percentage of the vertices while it contains 95 percentage of edges. This indicates that the biconnected core is dense while the detached graph is quite sparse. Recall that the biconnected core is one connected component. On the other hand, in the detached graph, there are many connected components, which implies that the vertices in the detached graph are likely to be disconnected with each other. Notice that each connected component in the detached graph corresponds to a whisker. So, the largest connected component can be interpreted as the largest whisker. Based on the statistics of the detached graph, we can see  that whiskers tend to be separable from each other, and there are no significant size whiskers. Also, the gap between the sizes of the biconnected core and the largest whisker is significant. All these statistics and observations support that our filtering phase creates a reasonable and more tractable input for an overlapping community detection algorithm.

\subsection{Seeding Phase}
\label{sec:seeding}
Once we obtain the biconnected core graph, we find seeds in this filtered graph. The goal of an effective seeding strategy is to identify a diversity of vertices, each of which lies within a cluster of good conductance. This identification should not be too computationally expensive. 


\textbf{Graclus Centers.}
One way to achieve these goals is to first apply a high quality and fast graph partitioning scheme (disjoint clustering of vertices in a graph) in order to compute a collection of sets with fairly small conductance. Then, we select a set of seeds by picking the most central vertex from each set (cluster). The idea here is roughly that we want something that is close to the partitioning -- which ought to be good -- but that allows overlap to produce better boundaries between the partitions. 

See Algorithm~\ref{graclus_centers} for the full procedure. In practice, we perform top-down hierarchical clustering using Graclus \cite{dhillon-2007-graclus} to get a large number of clusters. Then, we take the center of each cluster as a seed -- the center of a cluster is defined to be the vertex that is closest to the cluster centroid (as discussed in Section \ref{sec:wkkm}, we can quantify the distance between a vertex and a cluster centroid by using the kernel that underlies the relationship between kernel $k$-means and graph clustering); see steps 5 and 7 in Algorithm \ref{graclus_centers}. If there are several vertices whose distances are tied for the center of a cluster, we include all of them.

\begin{algorithm}[t]
\caption{Seeding by Graclus Centers}
\label{graclus_centers}
\begin{algorithmic}[1]
\floatname{algorithm}{Procedure}
\renewcommand{\algorithmicrequire}{\textbf{Input:}}
\renewcommand{\algorithmicensure}{\textbf{Output:}}
\REQUIRE graph $G$, the number of seeds $k$.
\ENSURE the seed set $\mathcal{S}$.
\STATE Compute exhaustive and non-overlapping clusters $\mathcal{C}_i$ ($i$=$1,...,k$) on $G$.
\STATE Initialize $\mathcal{S} = \emptyset$.
\FOR {each cluster $\mathcal{C}_i$}
	\FOR {each vertex $v \in \mathcal{C}_i$}
	\STATE Compute $\dist(v, \mathcal{C}_i)$ using (\ref{dist}).
	\ENDFOR
	\STATE $\mathcal{S} = \{ \underset{v}{\operatorname{argmin}}$ $\dist(v, \mathcal{C}_i) \} \cup \mathcal{S}$.
\ENDFOR
\end{algorithmic}
\end{algorithm}

\textbf{Spread Hubs.}
From another viewpoint, the goal is to select a set of well-distributed seeds in the graph, such that they will have high coverage after we expand the sets. We greedily choose an independent set of $k$ points in the graph by looking at vertices in order of decreasing degree. For this heuristic, we draw inspiration from the distance function~\eqref{dist}, which shows that the distance between a vertex and a cluster is inversely proportional to degree. Thus, high degree vertices are expected to have small distances to many other vertices. This also explains why we call the method \emph{spread hubs}. It also follows from the recent results in \cite{gleich-2012-neighborhoods}, \cite{whang-2012-gem}, \cite{kang-2011-caveman} which show that there should be good clusters around high degree vertices in power-law graphs with high clustering coefficients. We use an independent set in order to avoid picking seeds nearby each other. 

Our full procedure is described in Algorithm~\ref{spread_hubs}. In the beginning, all the vertices are unmarked. Until $k$ seeds are chosen, the following procedure is repeated: among unmarked vertices, the highest degree vertex is selected as a seed, and then the selected vertex and its neighbors are marked. As the algorithm proceeds exploring hubs in the network, if there are several vertices whose degrees are the same, we take an independent set of those that are unmarked. This step may result in more than $k$ seeds, however, the final number of returned seeds typically does not exceed the input $k$ too much because there usually are not too many high degree vertices.

\begin{algorithm}[t]
\caption{Seeding by Spread Hubs}
\label{spread_hubs}
\begin{algorithmic}[1]
\floatname{algorithm}{Procedure}
\renewcommand{\algorithmicrequire}{\textbf{Input:}}
\renewcommand{\algorithmicensure}{\textbf{Output:}}
\REQUIRE graph $G = (\mathcal{V}, \mathcal{E})$, the number of seeds $k$.
\ENSURE the seed set $\mathcal{S}$.
\STATE Initialize $\mathcal{S} = \emptyset$.
\STATE All vertices in $\mathcal{V}$ are unmarked.
\WHILE {$|\mathcal{S}| < k$}
	\STATE Let $\mathcal{T}$ be the set of unmarked vertices 
	 with max degree.
	\FOR {\textbf{each} $t \in \mathcal{T}$}
		\IF {$t$ is unmarked}
			\STATE $\mathcal{S} = \{t\} \cup \mathcal{S}$.
			\STATE Mark $t$ and its neighbors.
		\ENDIF
	\ENDFOR
\ENDWHILE
\end{algorithmic}
\end{algorithm}


\subsection{Seed Expansion Phase}
\label{sec:ppr}

Once we have a set of seed vertices, we wish to expand the clusters around those seeds. An effective technique for this task is using a personalized PageRank (PPR) vector~\cite{page1999-pagerank}, also known as a random-walk with restart \cite{pan-2004-cross}. A personalized PageRank vector is the stationary distribution of a random walk that, with probability $\alpha$ follows a step of a random walk and with probability $(1-\alpha)$ jumps back to a seed node. If there are multiple seed nodes, then the choice is usually uniformly random.  Thus, nodes close by the seed are more likely to be visited. 

Recently, such techniques have been shown to produce communities that best match communities found in real-world networks \cite{abrahao-2012-ground}. In fact, personalized PageRank vectors have close relationships to graph cuts and clustering methods. Andersen et al. \cite{andersen-2006-local} show that a particular algorithm to compute a personalized PageRank vector, followed by a sweep over all cuts induced by the vector, will identify a set of good conductance within the graph. They prove this via a ``localized Cheeger inequality'' that states, informally, that the set identified via this procedure has a conductance that is not too far away from the best conductance of any set containing that vertex. Also, Mahoney et al. \cite{mahoney-2012-local} show that personalized PageRank is, effectively, a seed-biased eigenvector of the Laplacian. They also show a limit to relate the personalized PageRank vectors to the Fiedler vector of a graph.

\begin{algorithm}[t]
\caption{Seed Expansion by PPR}
\label{alg:ppr-cluster}
\begin{algorithmic}[1]
\floatname{algorithm}{Procedure}
\renewcommand{\algorithmicrequire}{\textbf{Input:}}
\renewcommand{\algorithmicensure}{\textbf{Output:}}
\REQUIRE graph $G = (\mathcal{V}, \mathcal{E})$, a seed node $s \in \mathcal{S}$, PageRank link-following probability parameter $0 < \alpha < 1$, accuracy $\varepsilon > 0$
\ENSURE low conductance set $\mathcal{C}$
\STATE Set $\mathcal{T} = \{ s \} \cup \{ \text{neighbors of s} \}$
\STATE Initialize $x_v \!=\! 0$ for $v \in \mathcal{V}$
\STATE Initialize $r_v = 0$ for $v \in \mathcal{V} \setminus \mathcal{T}$, $r_v \!=\! 1/|\mathcal{T}|$ for $v \in \mathcal{T}$
\WHILE {any $r_v >  \deg(v) \varepsilon$}
\STATE Update $x_v = x_v + (1-\alpha) r_v$.
\STATE For each $(v,u) \in \mathcal{E}$, \\ \hspace{1em} update $r_u = r_u + \alpha r_v/(2 \deg(v))$
\STATE Update $r_v = \alpha r_v/2$
\ENDWHILE
\STATE Sort vertices by decreasing $x_v/\deg(v)$
\STATE For each prefix set of vertices in the sorted list, compute
the conductance of that set and set $\mathcal{C}$ to be the set that achieves the minimum.
\end{algorithmic}
\end{algorithm}

We briefly summarize the PPR-based seed expansion procedure in Algorithm \ref{alg:ppr-cluster} (each seed is expanded by this procedure). Please see Andersen et al. \cite{andersen-2006-local} for a full description of the algorithm. The high level idea of this expansion method is that given a set of restart nodes (denoted by $\mathcal{T}$ in Algorithm \ref{alg:ppr-cluster}), we first compute the PPR vector, examine nodes in order of highest to lowest PPR score, and then return the set that achieves the minimum conductance.

It is important to note that we can have multiple nodes in $\mathcal{T}$ (which corresponds to nonzero elements in the personalization vector in PPR), and indeed we use the entire vertex neighborhood of a seed node as the restart nodes (see step 1 in Algorithm \ref{alg:ppr-cluster}). Since we do not just use a singleton seed but also use its neighbors as the restart nodes in PPR, we call step 1 \textit{neighborhood inflation}. We empirically observed that this neighborhood inflation plays a critical role in producing low conductance communities. See Section \ref{sec:exp} for details. Recently, Gleich and Seshadhri \cite{gleich-2012-neighborhoods} have provided some theoretical justification for why neighborhood-inflated seeds may outperform a singleton seed in PPR expansion on many real-world networks.

Steps 2-8 are closely related to a coordinate descent optimization procedure \cite{bonchi-2012-katz} on the PageRank linear system. Although it may not be apparent from the procedure, this algorithm is remarkably efficient when combined with appropriate data structures.  The algorithm keeps two vectors of values for each vertex, ${\bf x}$ and ${\bf r}$. In a large graph, most of these values will remain zero on the vertices and hence, these need not be stored. Our implementation uses a hash table for the vectors ${\bf x}$ and ${\bf r}$. Consequently, the sorting step is only over a small fraction of the total vertices. Typically, we find this method takes only a few milliseconds, even for a large graph.

In the original PPR clustering scheme \cite{andersen-2006-local}, the PPR score is divided by the degree of each node (step 9) to remove bias towards high degree nodes. This step converts a PageRank vector, a left eigenvector of a Markov chain, into the right eigenvector of a Markov chain. Right eigenvectors are close relatives of the Fiedler vector of a graph, and so this degree normalization produces a vector that we call the \emph{Fiedler Personalized PageRank vector} because of this relationship. Fiedler vectors also satisfy Cheeger inequalities, just like the Fiedler Personalized PageRank vectors.  However, Kloumann and Kleinberg \cite{kloumann-2014-membership} recently reported that this degree normalization might slightly degrade the quality of the output clusters in terms of matching with ground-truth communities in some real-world networks. So, in our experiments, we also try using the PPR score which we just call \textit{PPR}. We compare the performance of the Fiedler PPR and PPR in Section \ref{sec:exp}.

In Algorithm \ref{alg:ppr-cluster}, there are two parameters which are related to PPR computation: $\alpha$ and $\varepsilon$. We follow standard practice for PPR clustering on an undirected graph and set $\alpha = 0.99$ \cite{leskovec-2009-natural}. This value yields results that are similar to those without damping, yet have bounded computational time. The parameter $\varepsilon$ is an accuracy parameter. As $\varepsilon \to 0$, the final vector solution ${\bf x}$ tends to the exact solution of the PageRank linear system.  When used for clustering, however, this parameter controls the effective \emph{size} of the final cluster. If $\varepsilon$ is large (about $10^{-2})$, then the output vector is inaccurate, incredibly sparse, and the resulting cluster is small.  If $\varepsilon$ is small, say $10^{-8}$, then the PageRank vector is accurate, nearly dense, and the resulting cluster may be large.  We thus run the PPR clustering scheme several times, with a range of accuracy parameters that are empirically designed to produce clusters with between 1 and 50,000 times the number of  edges in the initial seed set. The final community we select is the one with the best conductance score from these possibilities.  

\subsection{Propagation Phase}
Once we get the personalized PageRank communities on the biconnected core, we further expand each of the communities to the regions detached in the filtering phase. Our assignment procedure is straightforward: for each detached whisker connected via a bridge, we add that piece to all of the clusters that utilize the other vertex in the bridge. This procedure is described in Algorithm~\ref{assign}. In this way, each community $\mathcal{C}_i$ is expanded. 


\begin{algorithm}[t]
\caption{Propagation Procedure}
\label{assign}
\begin{algorithmic}[1]
\floatname{algorithm}{Procedure}
\renewcommand{\algorithmicrequire}{\textbf{Input:}}
\renewcommand{\algorithmicensure}{\textbf{Output:}}
\REQUIRE graph $G = (\mathcal{V}, \mathcal{E})$, biconnected core $G_C = (\mathcal{V}_C, \mathcal{E}_C)$, communities of $G_C: \mathcal{C}_i$ ($i=1,...,k$) $\in \mathcal{C}$.
\ENSURE communities of $G$.
\FOR {each $\mathcal{C}_i \in \mathcal{C}$}
	\STATE Detect bridges $\mathcal{E}_{B_i}$ attached to $\mathcal{C}_i$.
	\FOR {each $b_j \in \mathcal{E}_{B_i}$}
	\STATE Detect the whisker $w_j = (\mathcal{V}_j, \mathcal{E}_j)$ which is attached to $b_j$.
 	\STATE $\mathcal{C}_i = \mathcal{C}_i \cup \mathcal{V}_j$.
	\ENDFOR
\ENDFOR
\end{algorithmic}
\end{algorithm}

We now show that our propagation procedure only improves the quality of the final clustering result in terms of the normalized cut metric. To do this, we need to fix some notation. Let $\mathcal{E}_{B_i}$ be a set of bridges which are attached to $\mathcal{C}_i$, and $W_{\mathcal{C}_i}$ be a set of whiskers which are attached to the bridges, i.e., $W_{\mathcal{C}_i} = (\mathcal{V}_{W_i}, \mathcal{E}_{W_i})$ where 
\[ 
  w_j = (\mathcal{V}_j, \mathcal{E}_j) \in W_{\mathcal{C}_i}; \; \; 
  \mathcal{V}_{W_i} = \bigcup\limits_{\mathclap{w_j \in W_{\mathcal{C}_i}}} \mathcal{V}_j; \; 
  \text{ and } \;  
  \mathcal{E}_{W_i} = \bigcup\limits_{\mathclap{w_j \in W_{\mathcal{C}_i}}} \mathcal{E}_j. \]
Finally, let $\mathcal{C}_i'$ denote the expanded $\mathcal{C}_i$, where $|\mathcal{C}_i'| \geq |\mathcal{C}_i|$. Equality holds in this expression when there is no bridge attached to $\mathcal{C}_i$. When we expand $\mathcal{C}_i$ using Algorithm \ref{assign}, $\mathcal{C}_i'$ is equal to $\{ \mathcal{C}_i \bigcup \mathcal{V}_{W_i} \}$. The following results show that we only decrease the size of the (normalized) cut by adding the whiskers.

\begin{mythe}
If a community $\mathcal{C}_i$ is expanded to $\mathcal{C}_i'$ using Algorithm \ref{assign}, $\cut(\mathcal{C}_i') = \cut(\mathcal{C}_i) -  \links(\mathcal{V}_{W_i},\mathcal{C}_i)$. 
\label{th1}
\end{mythe}

\begin{proof}
Recall that $\cut(\mathcal{C}_i)$ is defined as follows:
\begin{align}
\cut(\mathcal{C}_i) &= \links(\mathcal{C}_i, \mathcal{V} \setminus \mathcal{C}_i). \nonumber \\
		 &= \links(\mathcal{C}_i, \mathcal{V}) - \links(\mathcal{C}_i, \mathcal{C}_i). \nonumber
\end{align}
Let us first consider $\links(\mathcal{C}_i', \mathcal{V})$ as follows: 
\begin{equation}
\links(\mathcal{C}_i', \mathcal{V}) = \links(\mathcal{C}_i, \mathcal{V}) + \links(\mathcal{V}_{W_i}, \mathcal{V}) \\ - \links(\mathcal{V}_{W_i}, \mathcal{C}_i). \nonumber
\end{equation}
Notice that $\links(\mathcal{V}_{W_i}, \mathcal{V}) = \links(\mathcal{V}_{W_i}, \mathcal{V}_{W_i}) + \links(\mathcal{V}_{W_i}, \mathcal{C}_i)$ by definition of whiskers. Thus, $\links(\mathcal{C}_i', \mathcal{V})$ can be expressed as follows:
\begin{equation}
\label{civ}
\links(\mathcal{C}_i', \mathcal{V}) = \links(\mathcal{C}_i, \mathcal{V}) + \links(\mathcal{V}_{W_i}, \mathcal{V}_{W_i}).
\end{equation}

On the other hand, $\links(\mathcal{C}_i', \mathcal{C}_i')$ can be expressed as follows:
\begin{equation}
\label{cici}
\links(\mathcal{C}_i', \mathcal{C}_i') = \links(\mathcal{V}_{W_i},\mathcal{V}_{W_i}) + \links(\mathcal{C}_i,\mathcal{C}_i) \\ + \links(\mathcal{V}_{W_i},\mathcal{C}_i). 
\end{equation}

Now, let us compute $\cut(\mathcal{C}_i')$ which is defined by
\begin{equation}
\label{cutci}
\cut(\mathcal{C}_i') = \links(\mathcal{C}_i', \mathcal{V}) - \links(\mathcal{C}_i', \mathcal{C}_i').
\end{equation}

By rewriting (\ref{cutci}) using (\ref{civ}) and (\ref{cici}), we can express $\cut(\mathcal{C}_i')$ as follows: \\
\centerline{$\cut(\mathcal{C}_i') = \cut(\mathcal{C}_i) - \links(\mathcal{V}_{W_i},\mathcal{C}_i)$.}
\end{proof}

\begin{mythe}
If a community $\mathcal{C}_i$ is expanded to $\mathcal{C}_i'$ using Algorithm \ref{assign}, $\ncut(\mathcal{C}_i') \leq \ncut(\mathcal{C}_i)$.
\end{mythe}

\begin{proof}
Recall that 
\begin{align}
\ncut(\mathcal{C}_i) &= \dfrac{\cut(\mathcal{C}_i)}{\links(\mathcal{C}_i, \mathcal{V})}. \nonumber
\end{align}
On the other hand, by Theorem \ref{th1}, we can represent $\ncut(\mathcal{C}_i')$ as follows:
\begin{align}
\ncut(\mathcal{C}_i') &= \dfrac{\cut(\mathcal{C}_i')}{\links(\mathcal{C}_i', \mathcal{V})}. \nonumber \\
&= \dfrac{\cut(\mathcal{C}_i)-  \links(\mathcal{V}_{W_i},\mathcal{C}_i)}{\links(\mathcal{C}_i, \mathcal{V}) + \links(\mathcal{V}_{W_i}, \mathcal{V}_{W_i})}. \nonumber
\end{align}
Therefore, $\ncut(\mathcal{C}_i') \leq \ncut(\mathcal{C}_i)$. Equality holds when there is no bridge attached to $\mathcal{C}_i$, i.e., $\mathcal{E}_{B_i} = \emptyset$.
\end{proof}

\subsection{Time Complexity Analysis}
We summarize the time complexity of our overall algorithm in Table \ref{time}. The filtering phase requires computing biconnected components in a graph, which takes $O(|\mathcal{V}| + |\mathcal{E}|)$ time. The complexity of ``Graclus centers'' seeding strategy is determined by the complexity of hierarchical clustering using Graclus. Recall that ``Spread hubs'' seeding strategy requires nodes to be sorted according to their degrees. Thus, the complexity of this strategy is bounded by the sorting operation (we can use a bucket sort). Expanding each seed requires solving multiple personalized PageRank clustering problems. The complexity of this operation is complicated to state compactly \cite{andersen-2006-local}, but it scales with the output size of each cluster, $\links(\mathcal{C}_i,\mathcal{V}_C)$. Finally, our simple propagation procedure scans the regions that were not included in the biconnected core and attaches them to the final communities.

\begin{table}[htbp]
    \caption{Time complexity of each phase.}
    \centering
    {\footnotesize
    \begin{tabularx}{0.6\linewidth}{lll}
    \toprule
     Phase & & Time complexity \\
    \midrule
    Filtering & & $O(|\mathcal{V}| + |\mathcal{E}|)$ \\
    \midrule
    \multirow{2}{*}{Seeding} & Graclus centers & {\footnotesize $O( \lceil \log k  \rceil ( |\mathcal{V}_C| + |\mathcal{E}_C| ) )$} \\ 
     & Spread hubs & {\footnotesize $O(|\mathcal{V}_C|)$} \\ 
    \midrule
    Seed expansion & & {\footnotesize $O(\sum_{i}^k \links(\mathcal{C}_i,\mathcal{V}_C))$} \\
    \midrule
    Propagation &  & {\footnotesize $O(\sum_{i}^k (\mathcal{E}_{B_i} + \mathcal{V}_{W_i}+ \mathcal{E}_{W_i}))$} \\
    \bottomrule
    \end{tabularx}
    }
 	\label{time}
\end{table}

\section{Related Work}
For overlapping community detection, many different approaches have been proposed \cite{xie-2012-survey} including clique percolation, line graph partitioning, eigenvector methods, ego network analysis, and low-rank models. Clique percolation methods look for overlap between fixed size cliques in the graph \cite{palla-2005-overlapping}. Line graph partitioning is also known as link communities. Given a graph $G=(\mathcal{V}, \mathcal{E})$, the line graph of $L(G)$ (also called the dual graph) has a vertex for each edge in $G$ and an edge whenever two edges (in $G$) share a vertex. For instance, the line graph of a star is a clique. A partitioning of the line graph induces an overlapping clustering in the original graph \cite{ahn-2010-line}. Even though these clique percolation and line graph partitioning methods are known to be useful for finding meaningful overlapping structures, these methods often fail to scale to large networks like those we consider. 


Eigenvector methods generalize spectral methods and use a soft clustering scheme applied to eigenvectors of the normalized Laplacian or modularity matrix in order to estimate communities \cite{zhang-2007-fuzzy}. Ego network analysis methods use the theory of structural holes \cite{burt-1995-structural}, and compute and combine many communities through manipulating ego networks \cite{rees-2010-collective}, \cite{coscia-2012-demon}. We compare against the Demon method \cite{coscia-2012-demon} that uses this strategy. We also note that other low-rank methods such as non-negative matrix factorizations identify overlapping communities as well. We compare against the Bigclam method \cite{yang-2013-bigclam} that uses this approach.

The approach we employ is called local optimization and expansion \cite{xie-2012-survey}. Starting from a seed, such a method greedily expands a community around that seed until it reaches a local optima of the community detection objective. Determining how to seed a local expansion method is, arguably, a critical problem within these methods. Strategies to do so include using maximal cliques \cite{shen-2009-eagle}, prior information \cite{gargi-2011-youtube}, or locally minimal neighborhoods \cite{gleich-2012-neighborhoods}. The latter method was shown to identify the vast majority of good conductance sets in a graph; however, there was no provision made for total coverage of all vertices.

Different optimization objectives and expansion methods can be used in a local expansion method. For example, Oslom \cite{lancichinetti-2011-oslom} tests the statistical significance of clusters with respect to a random configuration during community expansion. Starting from a randomly picked node, the Oslom method greedily expands the cluster by checking whether the expanded community is statistically significant or not, which results in detecting a set of overlapping clusters and  outliers in a graph. We compare our method with the Oslom method in our experiments (see Section \ref{sec:exp}).

In our algorithm, we use a personalized PageRank based cut finder \cite{andersen-2006-local} for the local expansion method. Abrahao et al. \cite{abrahao-2012-ground} observe that the structure of real-world communities can be well captured by the random-walk-based algorithms, i.e., personalized PageRank clusters are topologically similar to real-world clusters. More recently, Kloumann and Kleinberg \cite{kloumann-2014-membership} propose to use pure PageRank scores instead of the Fiedler PageRank scores to get a higher accuracy in terms of matching with ground-truth communities. 

\begin{table*}[t] 
{\scriptsize
\caption{Returned number of clusters and graph coverage of each algorithm}
\centering
\begin{tabularx}{\textwidth}{XXXXXXX}
\toprule
Graph & & oslom & demon & bigclam & nise-sph-fppr & nise-grc-fppr \\ \midrule
HepPh &  coverage (\%) & 100 & 88.83 & 84.37 & 100 & 100 \\ 
&  no. of clusters  & 608 & 5,147 & 100 & 99 & 90 \\ 
\midrule
AstroPh &  coverage (\%) & 100 & 94.15 & 91.11 & 100 & 100 \\ 
&  no. of clusters  & 1,241 & 8,259 & 200 & 212 & 246 \\ 
\midrule
CondMat &  coverage (\%) & 100 & 91.16 & 99.96 & 100 & 100 \\ 
&  no. of clusters  & 1,534 & 10,474 & 200 & 201 & 249 \\ 
\midrule
Flickr &  coverage (\%) & N/A & N/A & 52.13 & 93.60 & 100 \\  
&  no. of clusters  & N/A & N/A & 15,000 & 15,349 & 16,347 \\ 
\midrule
LiveJournal &  coverage (\%) & N/A & N/A & 43.86 & 99.78 & 99.79 \\ 
&  no. of clusters  & N/A & N/A & 15,000 & 15,058 & 16,271 \\ 
\midrule
Myspace &  coverage (\%) & N/A & N/A & N/A & 99.87 & 100 \\ 
&  no. of clusters  & N/A & N/A & N/A & 15,324 & 16,366 \\
\midrule
DBLP &  coverage (\%) & 100 & 84.89 & 100 & 100 & 100 \\ 
&  no. of clusters  & 17,519 & 174,560 & 25,000 & 26,503 & 18,477 \\ 
\midrule
Amazon &  coverage (\%) & 100 & 79.16 & 100 & 100 & 100 \\ 
&  no. of clusters  & 17,082 & 105,685 & 25,000 & 27,763 & 20,036 \\ 
\midrule
Orkut &  coverage (\%) & N/A & N/A & 82.13 & 99.99 & 100 \\ 
&  no. of clusters  & N/A & N/A & 25,000 & 25,204 & 32,622 \\
\midrule
LiveJournal2 &  coverage (\%) & N/A & N/A & 56.64 & 99.95 & 99.99 \\ 
&  no. of clusters  & N/A & N/A & 25,000 & 25,065 & 32,274 \\   
\bottomrule
\end{tabularx}
\label{cov_table}
}
\end{table*}

A preliminary version of this work has appeared in \cite{whang-2003-sse}. In this paper, we provide technical details about neighborhood inflation in our seed expansion phase, and include additional experimental results to show the importance of the neighborhood inflation step. Also, we test and compare the performance of the Fiedler PageRank and the standard PPR in our expansion phase. We also improve the implementation of our algorithm in that we try expanding seeds in parallel using multiple threads.

\section{Experimental Results}
\label{sec:exp}
\vspace{-1.1cm}
\begin{center}
\begin{figure*}[ht]
\begin{minipage}[b]{1\linewidth}\centering
  \centering
  \begin{tabular}{ccccc}
    \subfloat[HepPh]{\includegraphics[width=0.2\textwidth]{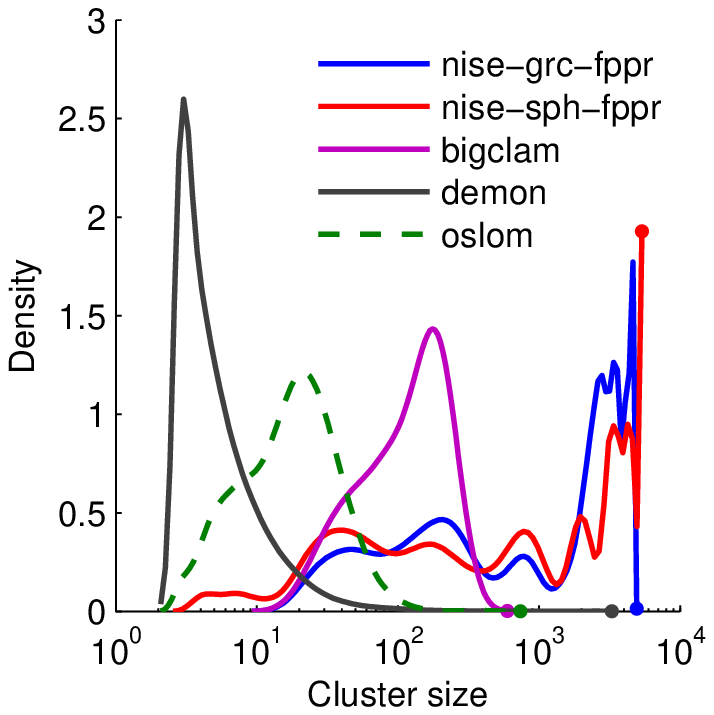}}
    \subfloat[AstroPh]{\includegraphics[width=0.2\textwidth]{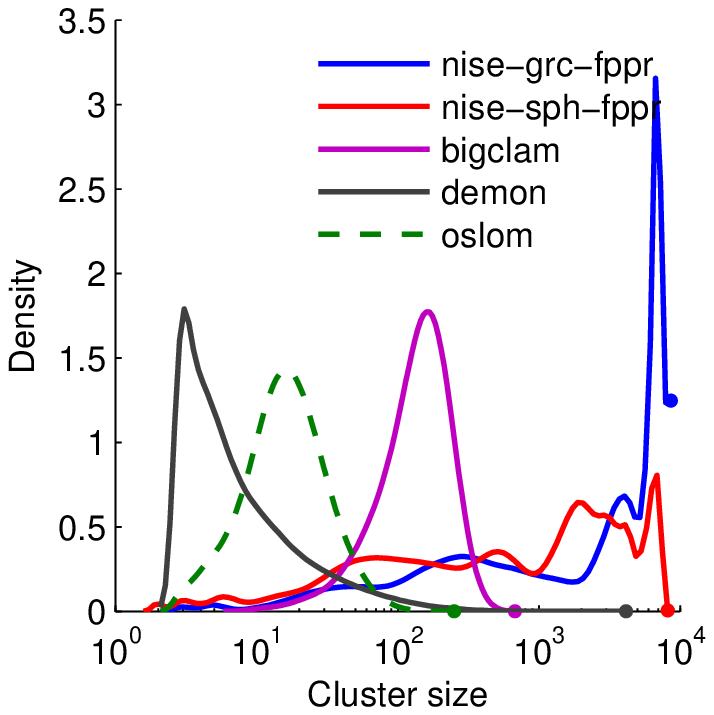}} 
    \subfloat[CondMat]{\includegraphics[width=0.2\textwidth]{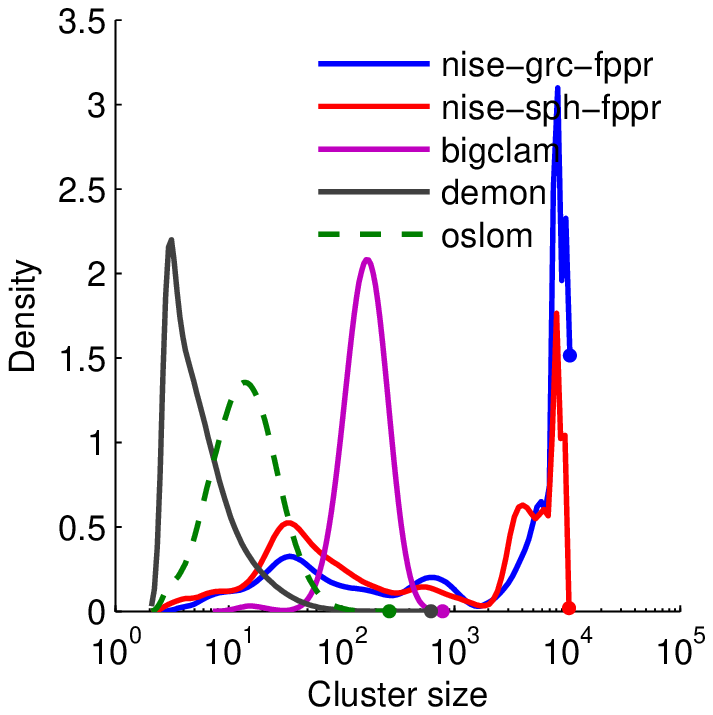}}
    \subfloat[DBLP]{\includegraphics[width=0.2\textwidth]{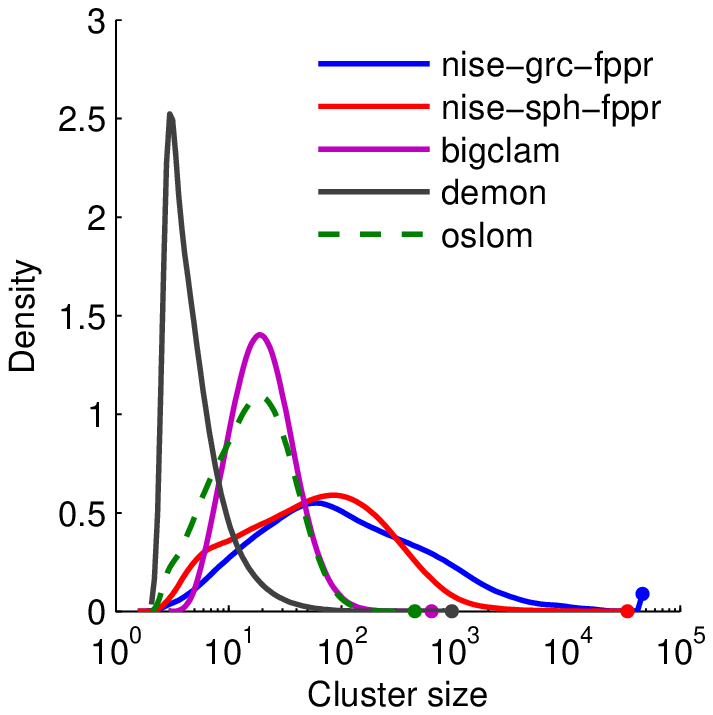}}
    \subfloat[Amazon]{\includegraphics[width=0.2\textwidth]{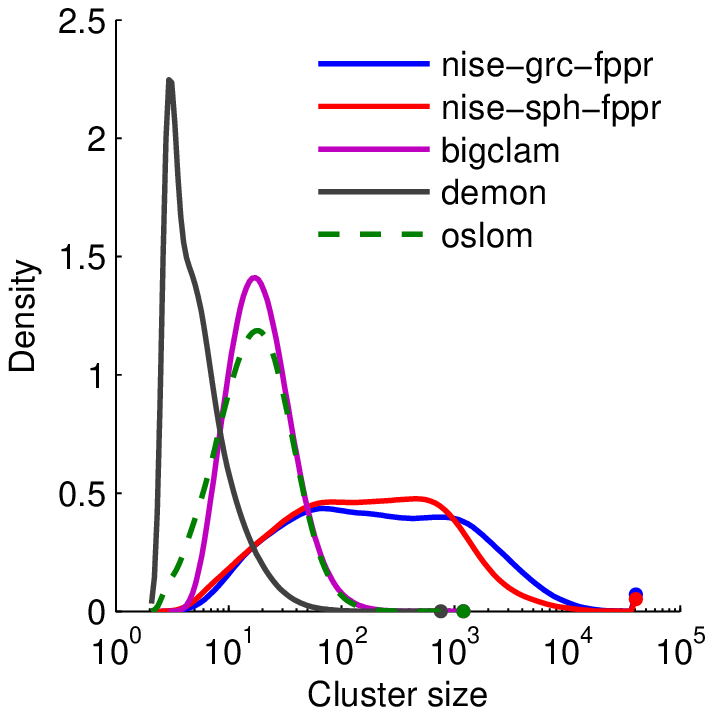}} 
  \end{tabular}
\end{minipage} 
\begin{minipage}[b]{1\linewidth}\centering
  \centering
  \begin{tabular}{ccccc}
    \subfloat[Orkut]{\includegraphics[width=0.2\textwidth]{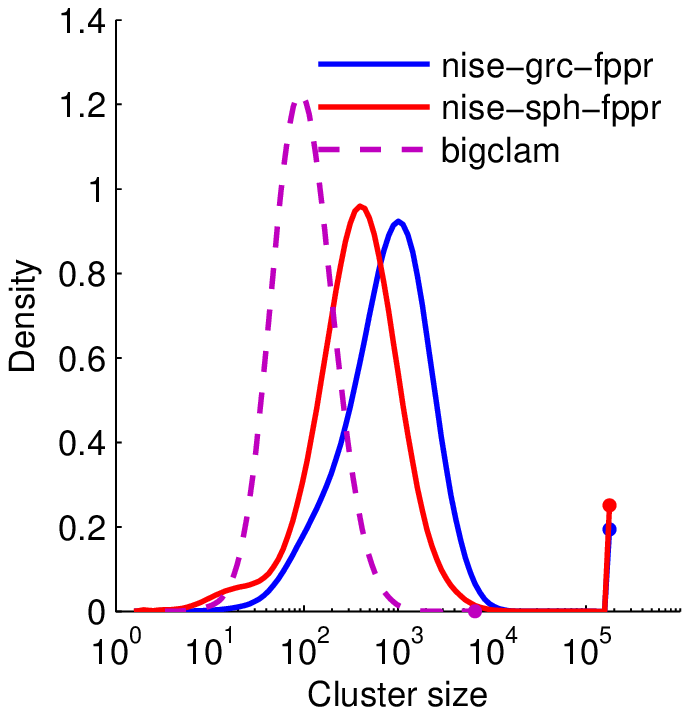}}
    \subfloat[Flickr]{\includegraphics[width=0.2\textwidth]{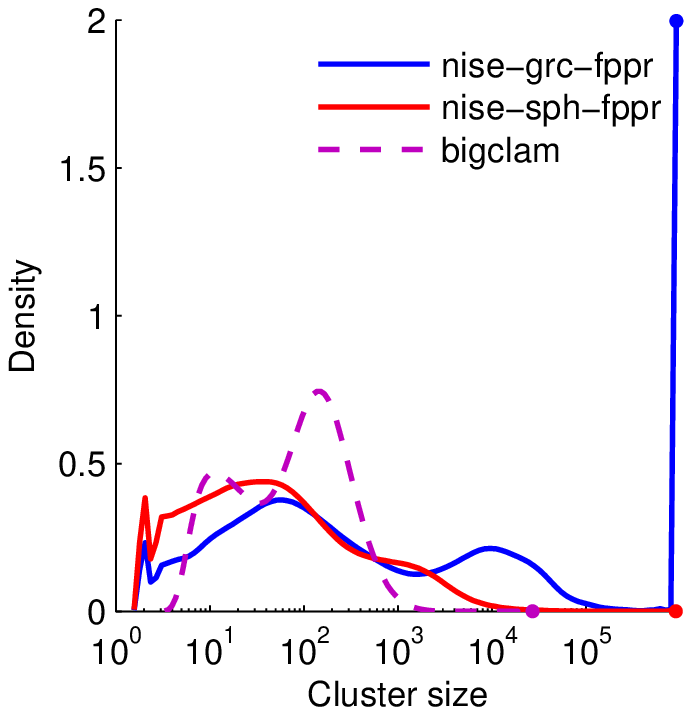}}
    \subfloat[LiveJournal]{\includegraphics[width=0.2\textwidth]{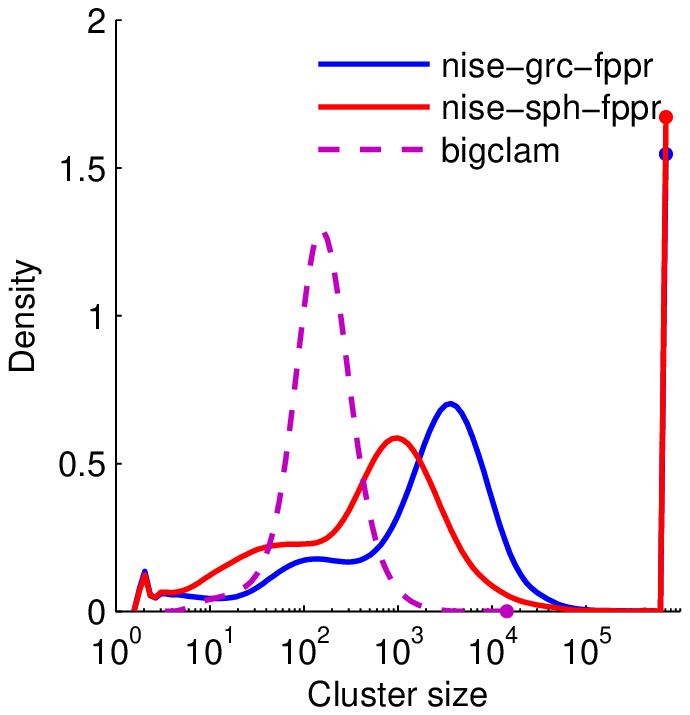}} 
    \subfloat[LiveJournal2]{\includegraphics[width=0.2\textwidth]{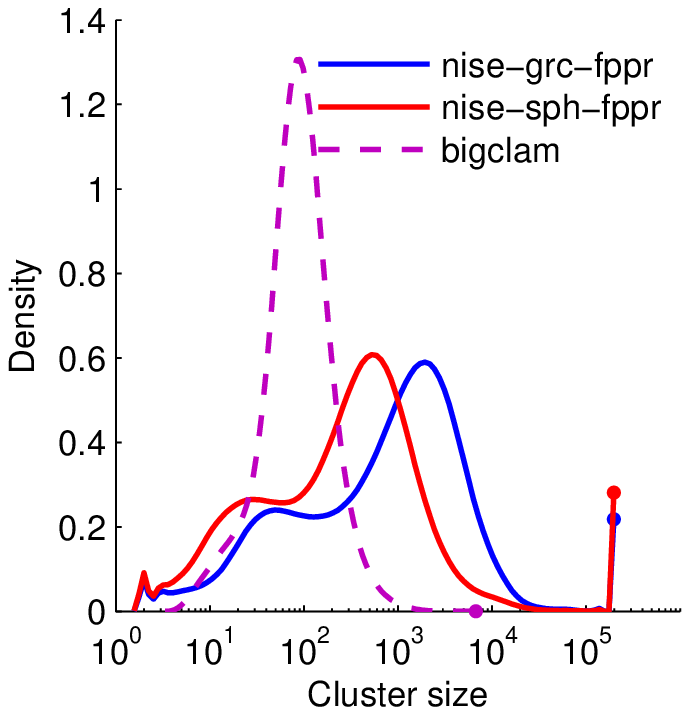}}
    \subfloat[Myspace]{\includegraphics[width=0.2\textwidth]{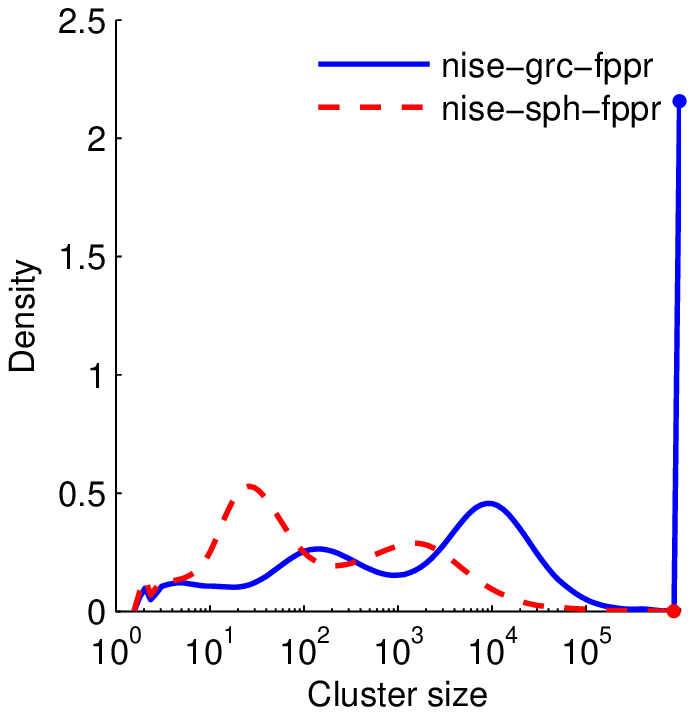}} 
  \end{tabular}
\end{minipage} 
\caption{Distributions of cluster sizes from the methods. These plots show a kernel density smoothed histogram of the cluster sizes from each method. The horizontal axis is the cluster size and the vertical axis is proportional to the number of clusters of that size.}
\label{csizes}
\end{figure*}
\end{center}

We compare our algorithm, \NISE, with other state-of-the-art overlapping community detection methods: Bigclam \cite{yang-2013-bigclam}, Demon \cite{coscia-2012-demon}, and Oslom \cite{lancichinetti-2011-oslom}. For these three methods, we used the software which is provided by the authors of \cite{yang-2013-bigclam}, \cite{coscia-2012-demon}, and \cite{lancichinetti-2011-oslom} respectively. While Demon and Oslom only support a sequential execution, Bigclam supports a multi-threaded execution. \NISE is written in a mixture of C++ and MATLAB. In \NISE, seeds can be expanded in parallel, and this feature is implemented using parallel computing toolbox provided by MATLAB. We compare the performance of each of these methods on ten different real-world networks which are presented in Section \ref{sec:data}.

Within \NISE, we also compare the performance of different seeding strategies and some variants of expansion methods. We use four different seeding strategies: ``graclus centers'' (denoted by ``nise-grc-*'') and ``spread hubs'' (denoted by ``nise-sph-*'') which are proposed in this manuscript, ``locally minimal neighborhoods'' (denoted by ``nise-lcm-*'') which has been proposed in \cite{gleich-2012-neighborhoods}, and random seeding strategy (denoted by ``nise-rnd-*'') where we randomly take $k$ seeds. Andersen and Lang \cite{andersen-2006-seed} have provided some theoretical justification for why random seeding also should be competitive. On the other hand, we also compare two different expansion methods: the Fiedler Personalized PageRank (denoted by ``nise-*-fppr''), and the standard Personalized PageRank (denoted by ``nise-*-ppr''). 

\subsection{Graph Coverage and Community Sizes}
We first report the returned number of clusters and the graph coverage of each algorithm in Table~\ref{cov_table}. The graph coverage indicates how many vertices are assigned to clusters (i.e., the number of assigned vertices divided by the total number of vertices in a graph). Note that we can control the number of seeds $k$ in \NISE and the number of clusters $k$ in Bigclam. We set $k$ (in our methods and Bigclam) as 100 for HepPh, 200 for AstroPh and CondMat, 15,000 for Flickr, Myspace, and LiveJournal, and 25,000 for DBLP, Amazon, LiveJournal2, and Orkut networks without any tuning and using the guidance that larger graphs can have more clusters. For the networks where we have ground-truth communities, we slightly overestimate the number of clusters $k$ since there usually exists a large number of ground-truth communities. Since we remove duplicate clusters after the PageRank expansion in \NISE, the returned number of clusters can be smaller than $k$. Also, since we choose all the tied seeds in ``graclus centers'' and ``spread hubs'', the returned number of clusters of these algorithms can be slightly larger than $k$. Recall that we use a top-down hierarchical clustering scheme in the ``graclus centers'' strategy. So, in this case, the returned number of clusters before filtering the duplicate clusters is slightly greater than or equal to $2^{\lceil \log k  \rceil}$. On the other hand, Demon and Oslom determine the number of clusters based on datasets. Demon and Oslom fail on Flickr, Myspace, LiveJournal, LiveJournal2, and Orkut. Bigclam does not finish on the Myspace network (using 4 threads) after running for one week.

Figure \ref{csizes} shows distributions of cluster sizes. These figures show that the \NISE method tends to find larger clusters than the other methods, usually about 10 to 100 times as large. Also, the \NISE method often finds a number of large clusters---these are the spikes on the right for subfigures (f)--(j). This tends to happen slightly more often for the ``graclus centers'' seeding strategy. The other observation is that \NISE tends to produce more variance in the sizes of the clusters than the other methods and the resulting histograms are not as sharply peaked.

\begin{center}
\begin{figure*}[t]
\centering
  \begin{tabular}{ccc}
    \subfloat[LiveJournal]{\includegraphics[width=0.32\textwidth]{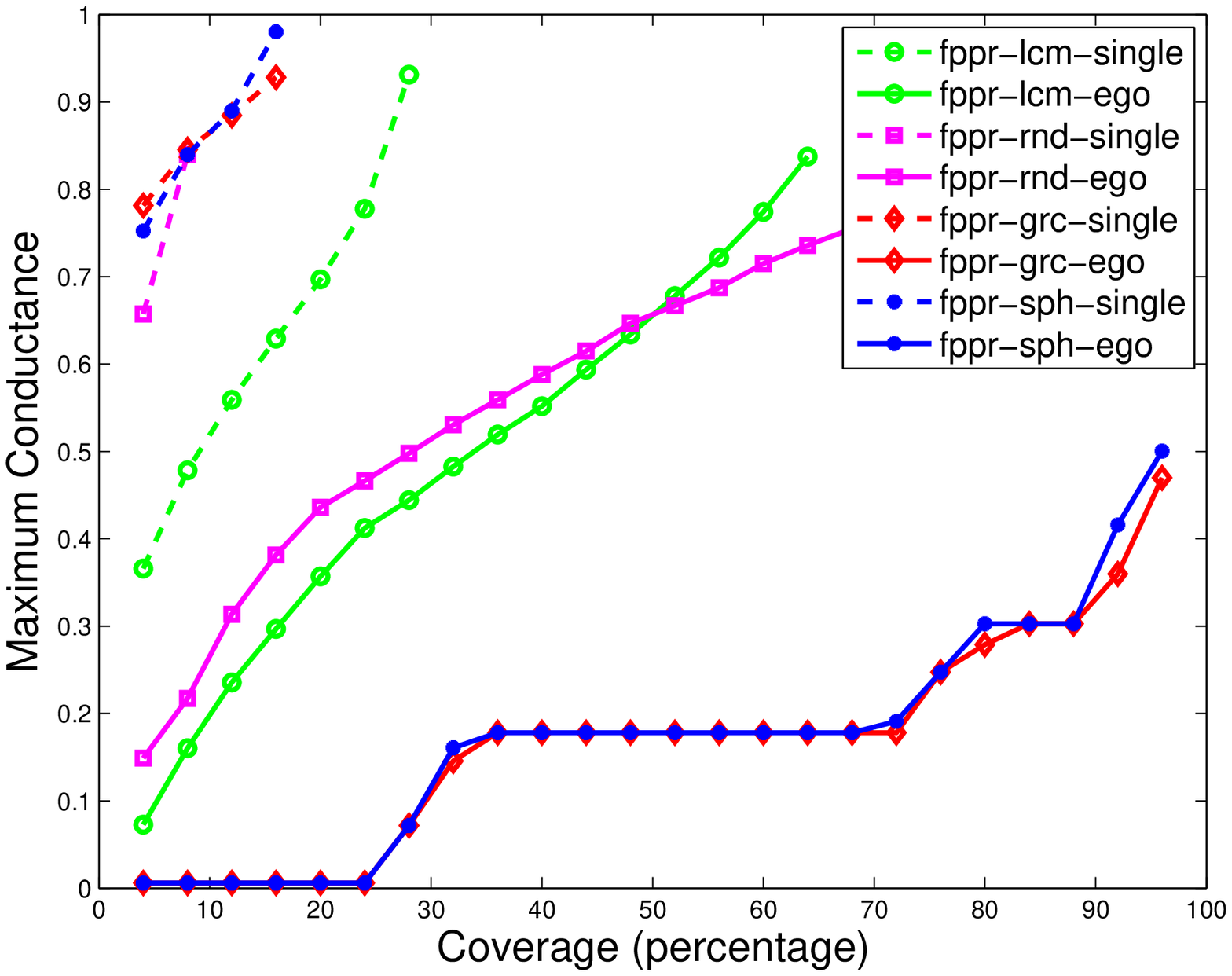}}&
    \subfloat[Myspace]{\includegraphics[width=0.32\textwidth]{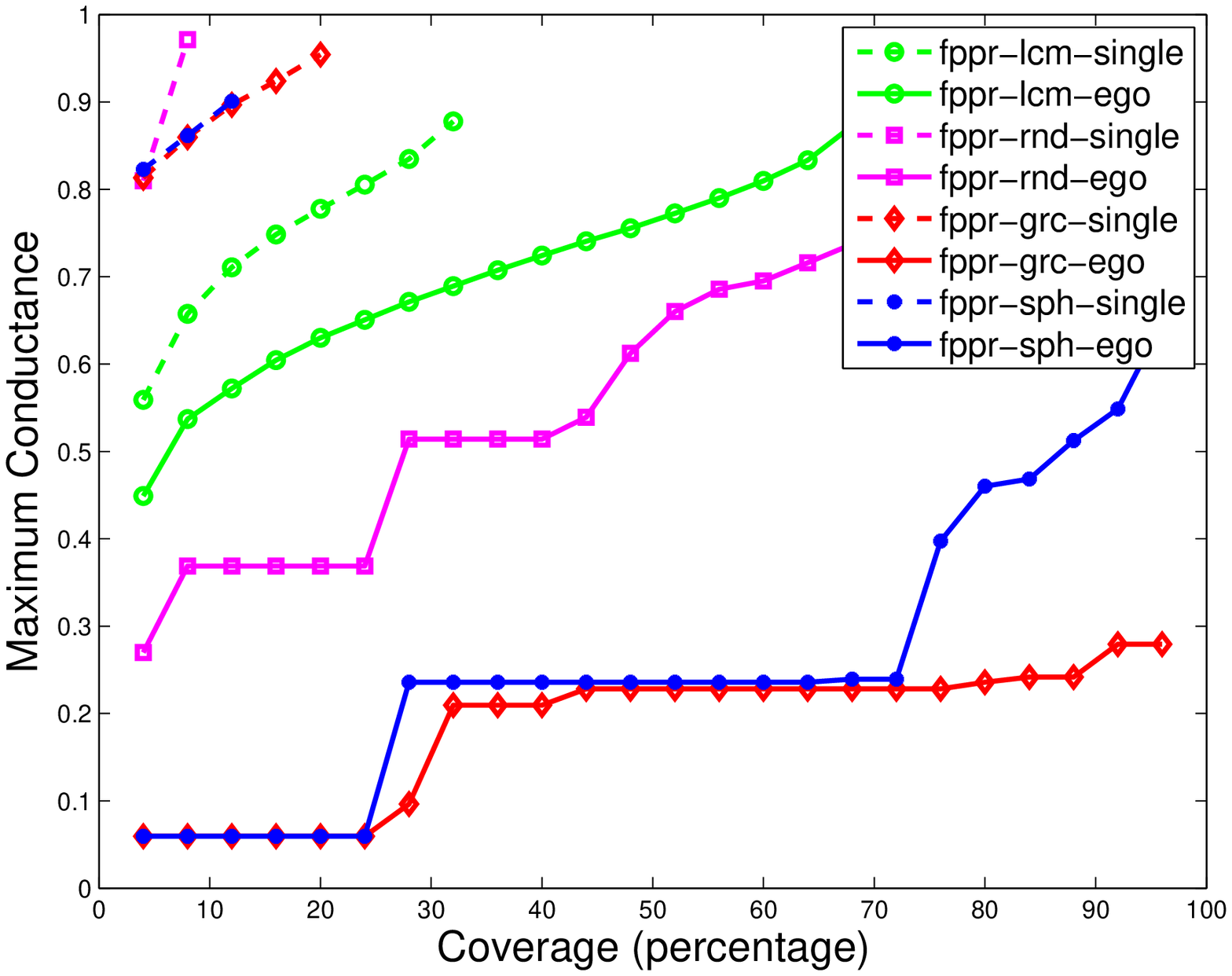}}
     \subfloat[Flickr]{\includegraphics[width=0.32\textwidth]{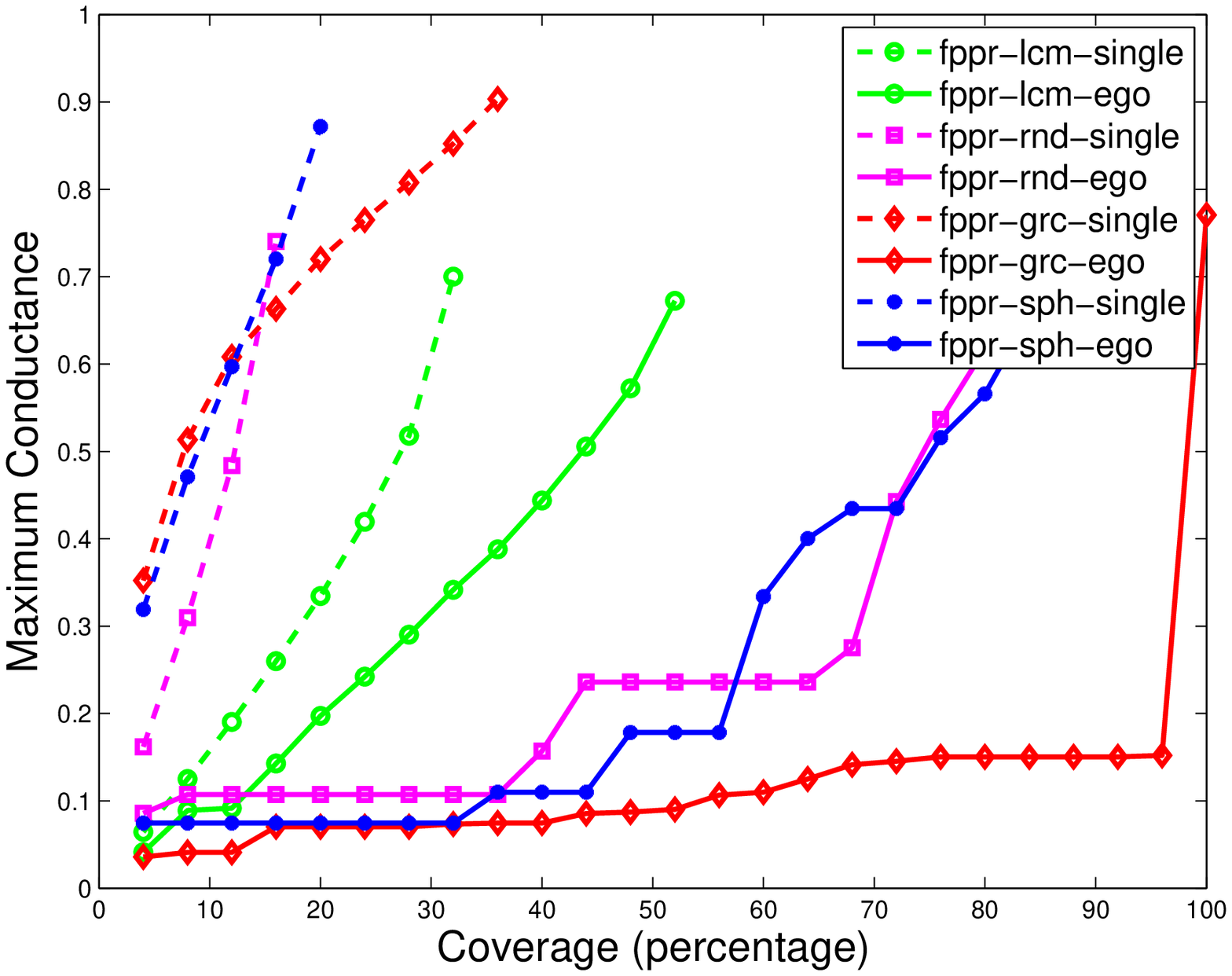}}&
  \end{tabular}
\caption{Importance of neighborhood inflation -- there is a large performance gap between singleton seeds and neighborhood-inflated seeds for all the seeding strategies. Neighborhood inflation plays a critical role in the success of \NISE. When neighborhood-inflated seeds are used, ``graclus centers'' and ``spread hubs'' seeding strategies significantly outperform other seeding strategies.}
\label{neigh_infl}
\end{figure*}
\end{center}

\vspace{-1cm}

\subsection{Importance of Neighborhood-Inflation}
We evaluate the quality of overlapping clustering in terms of the maximum conductance of any cluster. A high quality algorithm should return a set of clusters that covers a large portion of the graph with small maximum conductance. This metric can be presented by a conductance-vs-coverage curve. That is, for each method, we first sort the clusters according to the conductance scores in ascending order, and then greedily take clusters until a certain percentage of the graph is covered. The $x$-axis of each plot is the graph coverage, and the $y$-axis is the maximum conductance value among the clusters we take. We can interpret this plot as follows: we need to use clusters whose conductance scores are less than or equal to $y$ to cover $x$ percentage of the graph. Note that lower conductance indicates better quality of clusters, i.e., a lower curve indicates better clusters. 

First, we verify the importance of \textit{neighborhood inflation} in our seed expansion phase. Recall that when we compute the personalized PageRank (PPR) score for each seed node, we use the seed node's entire vertex neighborhood (the vertex neighborhood is also referred to as ``ego network'') as the restart region in PPR (details are in Section \ref{sec:ppr}). To see how this affects the overall performance of the seed expansion method, we compare the performance of singleton seeds and neighborhood-inflated seeds. Figure~\ref{neigh_infl} shows the conductance-vs-coverage plot for singleton seeds and neighborhood-inflated seeds. ``*-single'' indicates singleton seeds, i.e., each seed is solely used as the restart region in PPR. ``*-ego'' indicates neighborhood-inflated seeds. We also used four different seeding strategies: ``graclus centers'' (denoted by ``grc-*''), ``spread hubs'' (denoted by ``sph-*''), ``locally minimal neighborhoods'' (denoted by ``lcm-*''), and ``random'' (denoted by ``rnd-*''). 

We can see that the performance significantly degrades when singleton seeds are used for all the seeding strategies. This implies that neighborhood inflation plays a critical role in the success of our method. Even though we only present the results on LiveJournal, Myspace, and Flickr in Figure~\ref{neigh_infl} for brevity, we consistently observed that neighborhood-inflated seeds are much better than singleton seeds on all other networks. We also notice that that when neighborhood-inflated seeds are used, both ``graclus centers'' and ``spread hubs'' seeding strategies significantly outperform other seeding strategies. ``spread hubs'' and ``graclus centers'' seeding strategies produce similar results on LiveJournal whereas ``graclus centers'' is better than ``spread hubs'' on Myspace and Flickr. We used the conventional Fiedler PPR for the expansion phase in Figure~\ref{neigh_infl}, but we also got the same conclusion using the standard PPR.

\subsection{Community Quality Using Conductance}
We compute AUC (Area Under the Curve) of the conductance-vs-coverage to compare the performance of \NISE with other state-of-the-art methods. Within \NISE, we also compare four different seeding strategies and two different expansion methods. The AUC scores are normalized such that they are between zero and one. 

Figure~\ref{cond_plot} shows AUC scores on the six networks where we do not have ground-truth community information (see Table \ref{input} for details about these networks). We can see several patterns in Figure~\ref{cond_plot}. First, within \NISE, ``graclus centers'' and ``spread hubs'' seeding strategies outperform the other two seeding strategies. Second, for most of the cases, ``fppr'' leads to slightly better communities than ``ppr''. Also, we can see that ``nise-grc-fppr'' shows the best performance for all networks. Third, \NISE outperforms Demon, Oslom, and Bigclam. There is a significant performance gap between \NISE and these methods.

\begin{center}
\begin{figure*}[t]
\centering
\begin{minipage}[b]{1\linewidth}\centering
  \centering
  \begin{tabular}{ccc}
    \subfloat[AstroPh]{\includegraphics[width=0.33\textwidth]{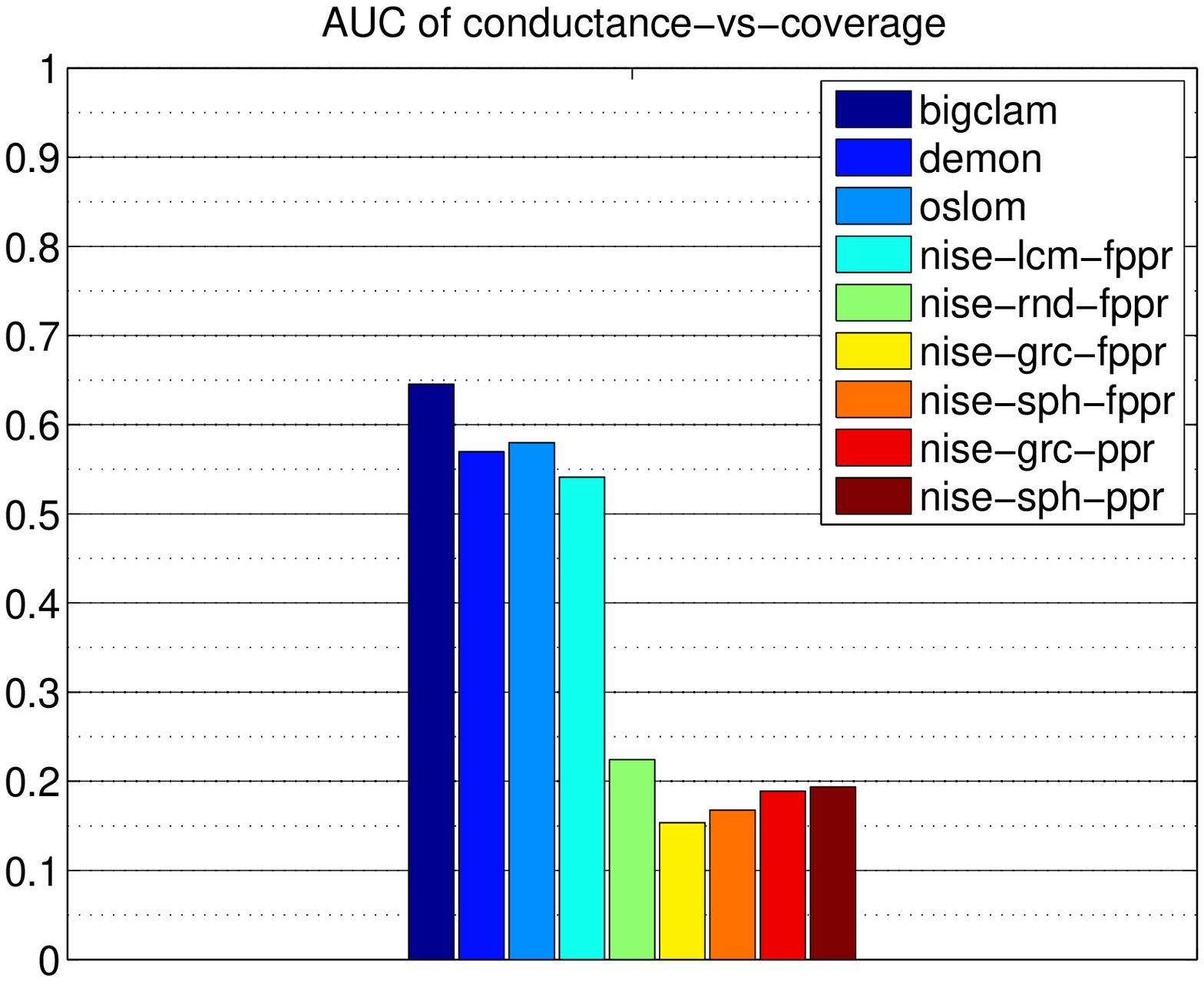}}
    \subfloat[HepPh]{\includegraphics[width=0.33\textwidth]{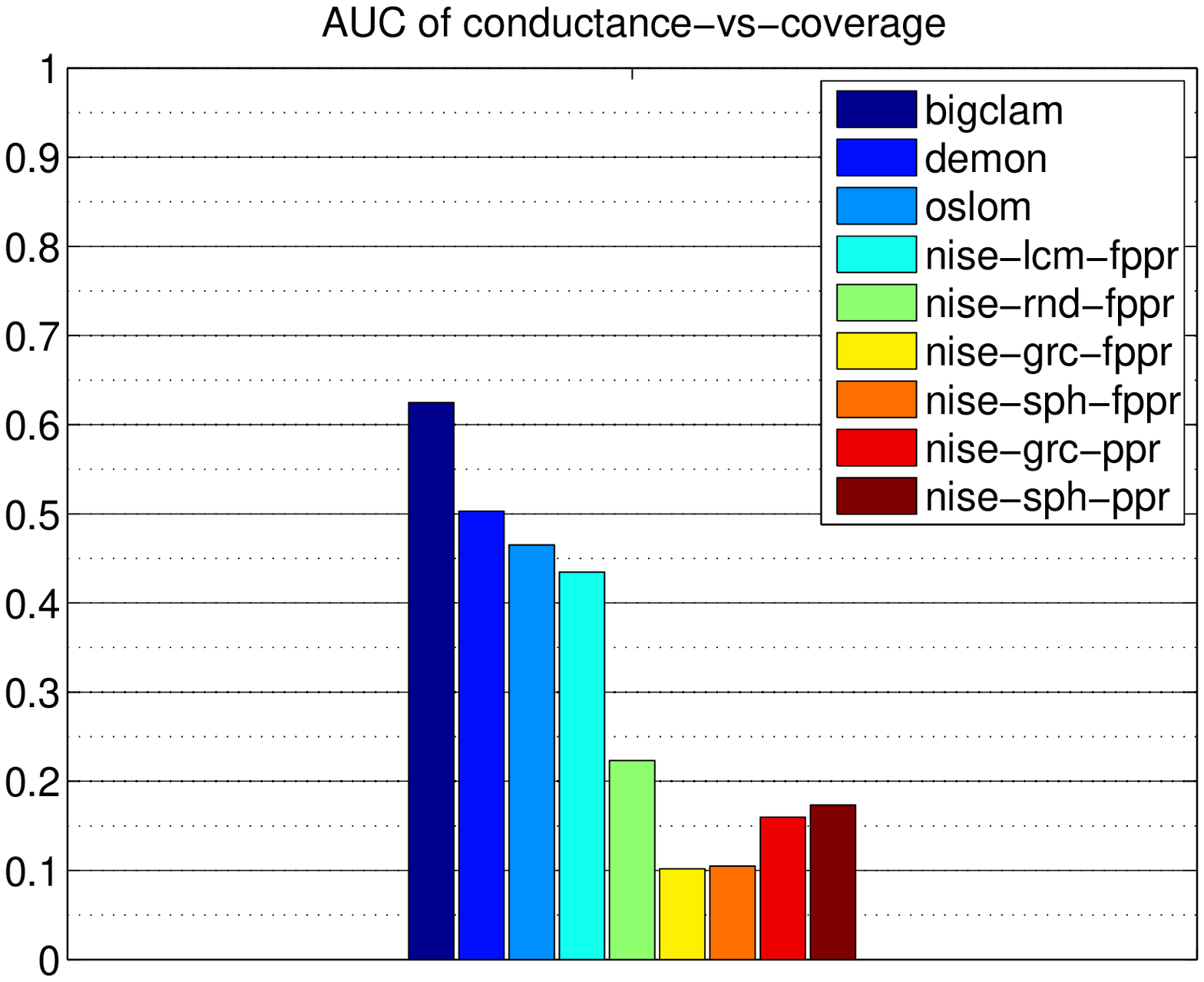}}
    \subfloat[CondMat]{\includegraphics[width=0.33\textwidth]{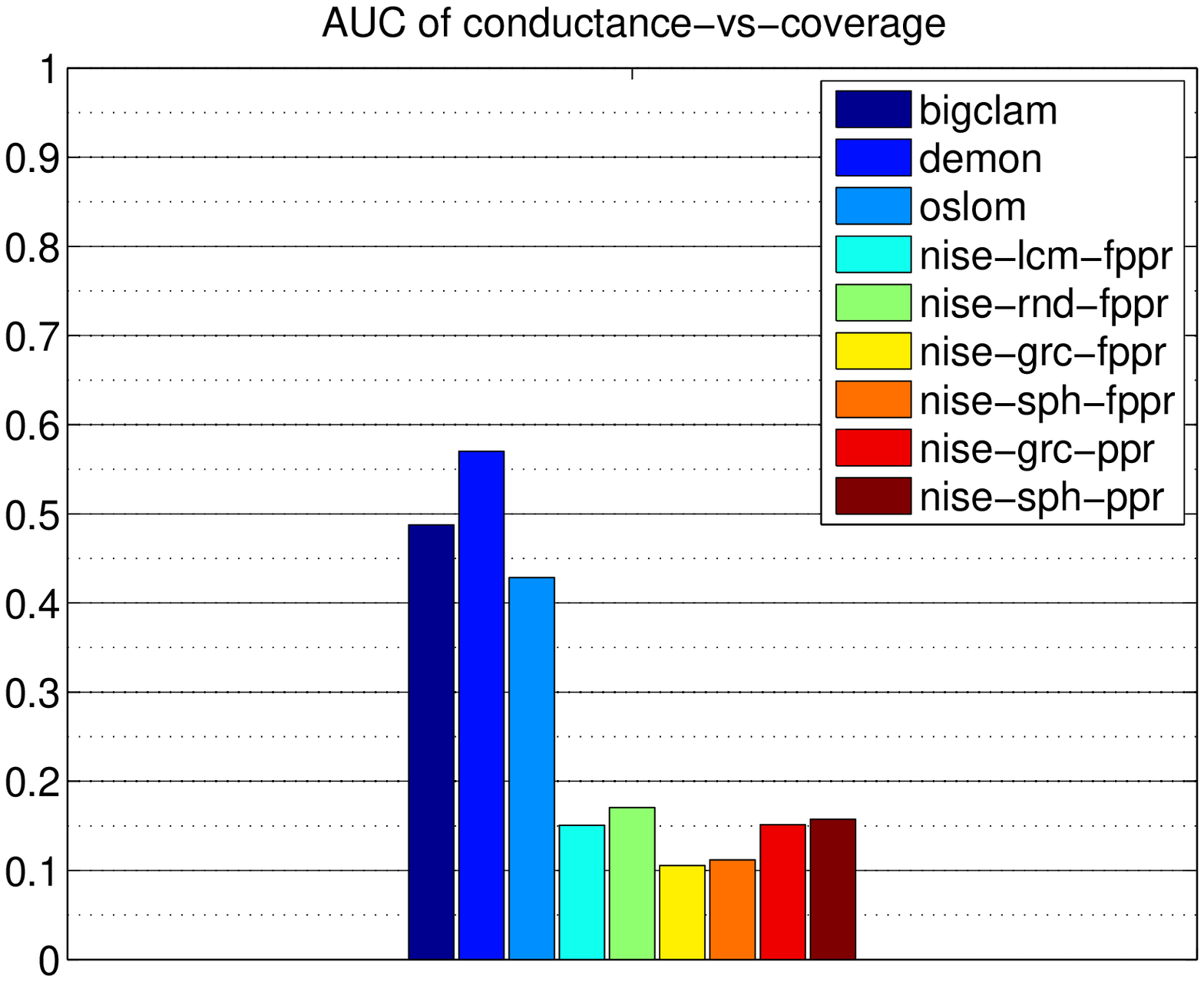}}
  \end{tabular}
\end{minipage}
\begin{minipage}[b]{1\linewidth}\centering
  \centering
  \begin{tabular}{ccc}  	
    \subfloat[Flickr]{\includegraphics[width=0.33\textwidth]{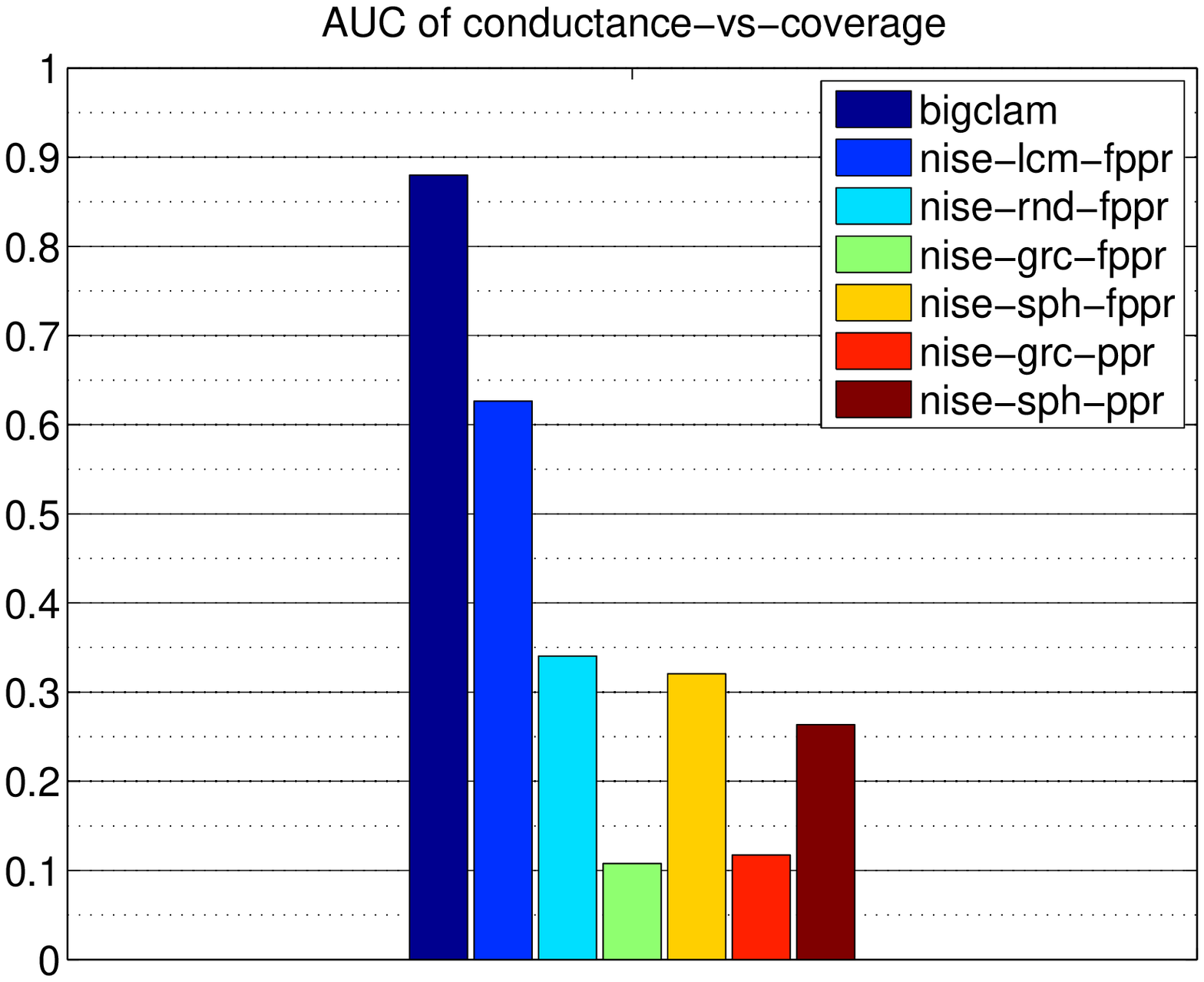}} 
    \subfloat[LiveJournal]{\includegraphics[width=0.33\textwidth]{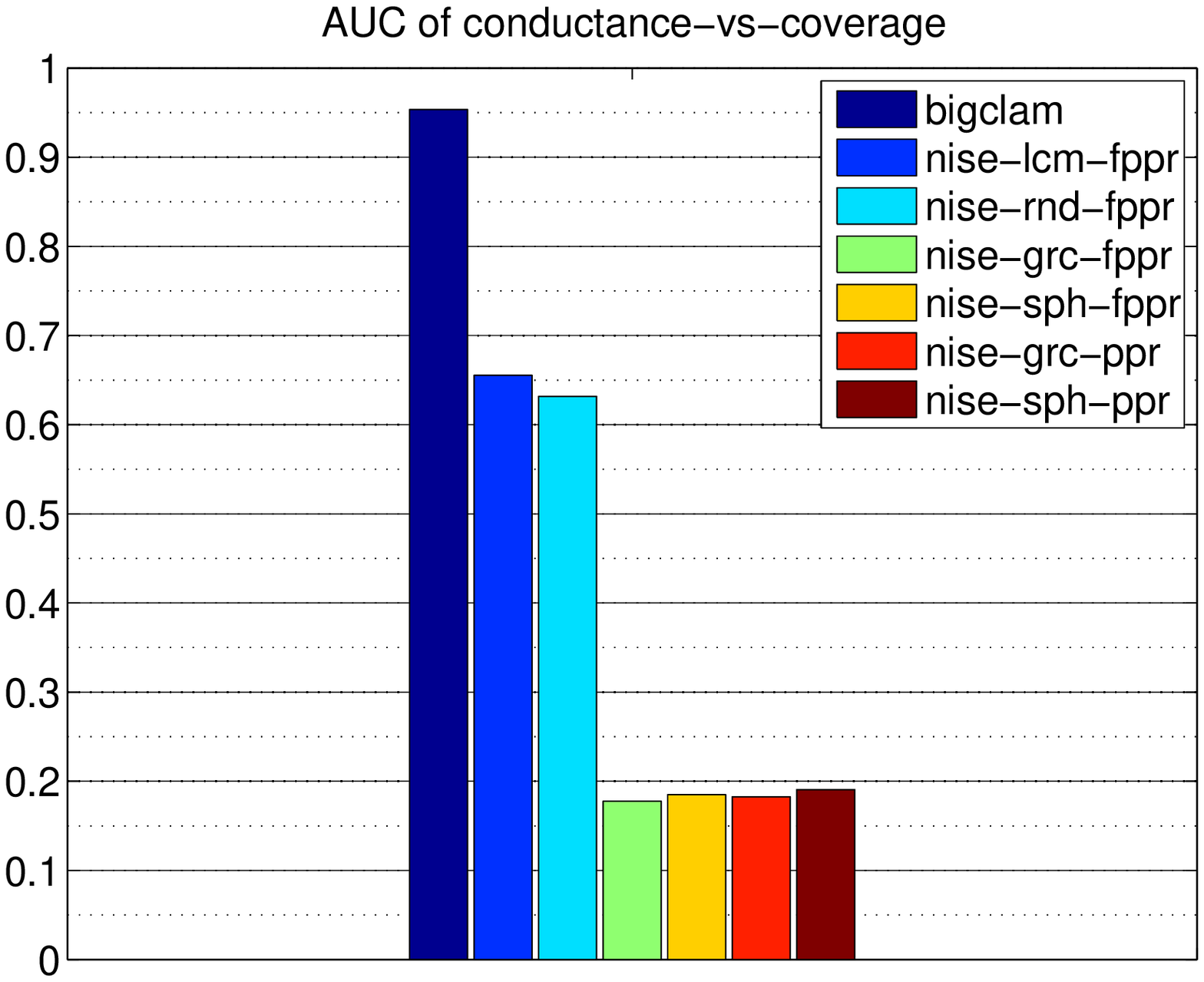}}
    \subfloat[Myspace]{\includegraphics[width=0.33\textwidth]{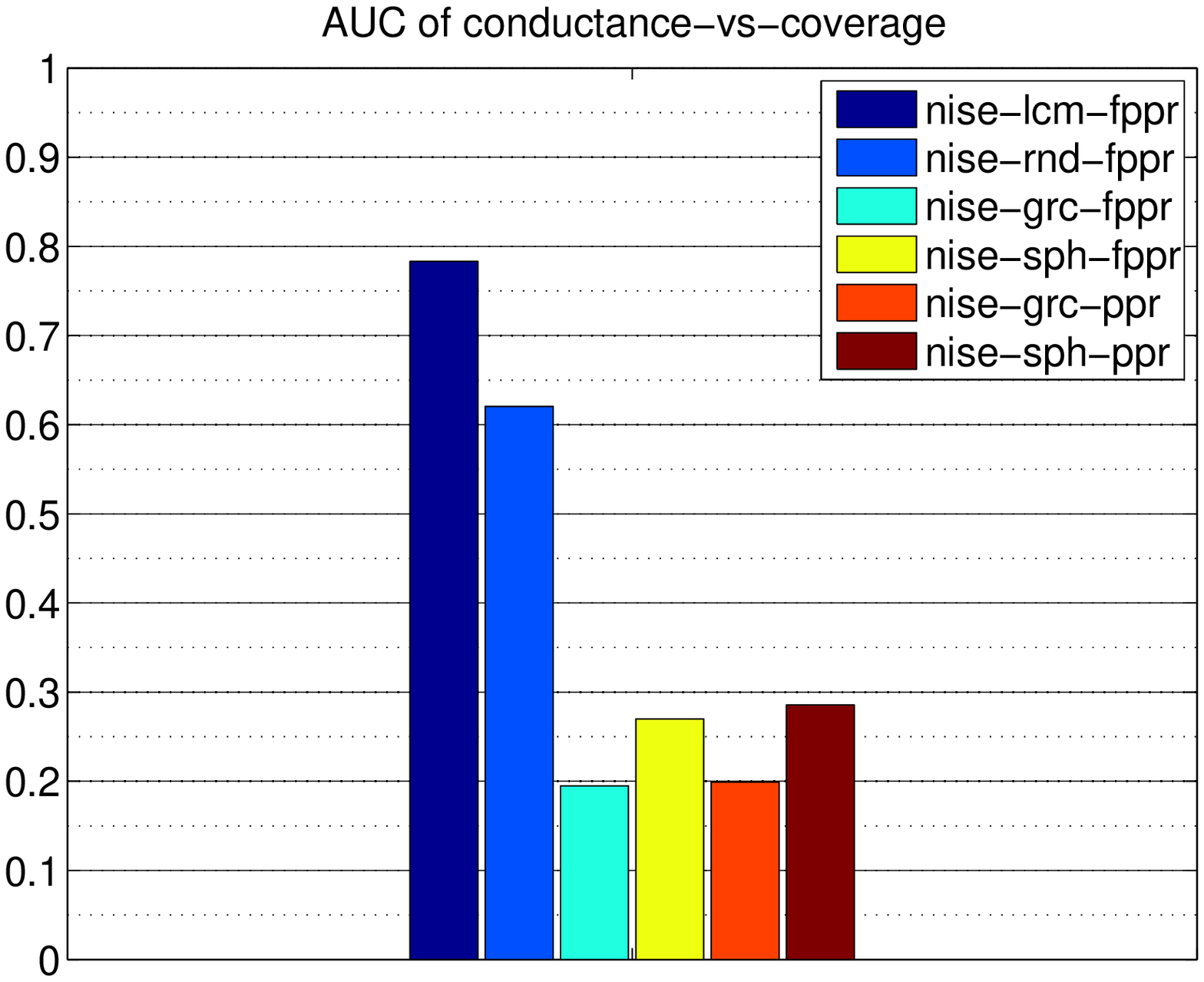}}
  \end{tabular}
\end{minipage}
\caption{AUC of Conductance-vs-coverage -- lower bar indicates better communities. \NISE outperforms Demon, Oslom, and Bigclam. Within \NISE, ``graclus centers'' and ``spread hubs'' seeding strategies are better than other seeding strategies, and the Fiedler PPR produces slightly better communities than the standard PPR.}
\label{cond_plot}
\end{figure*}
\end{center}

\begin{table*}[t] 
{\small
    \caption{F1 and F2 measures.} 
    \centering
    \begin{tabularx}{\linewidth}{lXXXXXXXX}
    \toprule
      & \multicolumn{2}{X}{DBLP} & \multicolumn{2}{X}{Amazon} & \multicolumn{2}{X}{LiveJournal2} & \multicolumn{2}{X}{Orkut} \\ 
      & $F_1$ & $F_2$ & $F_1$ & $F_2$ & $F_1$ & $F_2$ & $F_1$ & $F_2$ \\ \hline
    bigclam & 15.1 \% & 13.0 \% & 27.1 \% & 25.6 \% & 11.3 \% & 13.7 \% & 43.0 \% & 47.4 \% \\
    demon & 13.7 \% & 12.0 \% & 16.5 \% & 15.3 \% & N/A & N/A & N/A & N/A \\
    oslom & 13.4 \% & 11.6 \% & 32.0 \% & 30.2 \% & N/A & N/A & N/A & N/A \\
    nise-lcm-fppr & 13.9 \% & 15.4 \% & 46.3 \% & 56.5 \% & 11.3 \% & 13.8 \% & 40.9 \% & 46.8 \% \\
    nise-rnd-fppr & 17.7 \% & 20.5 \% & 48.9 \% & 58.8 \% & 12.1 \% & 16.5 \% & 54.6 \% & 62.9 \% \\
    nise-sph-fppr & 18.1 \% & 21.4 \% & 49.2 \% & {\bf 59.5} \% & 12.7 \% & {\bf 18.1} \% & 55.1 \% & 64.2 \% \\
    nise-sph-ppr & {\bf 19.0 \%} & {\bf 22.6 \%} & {\bf 49.7} \% & 58.7 \% & {\bf 12.8} \% & {\bf 18.1} \% & {\bf 57.4} \% & {\bf 65.2 \%} \\
    nise-grc-fppr & 17.6 \% & 21.7 \% & 46.7 \% & 57.1 \% & 12.2 \% & 17.6 \% & 51.1 \% & 61.4 \% \\
    nise-grc-ppr & 17.6 \% & 22.0 \% & 47.3 \% & 56.0 \% & {\bf 12.8} \% & 17.6 \% & 53.5 \% & 62.4 \% \\    
    \bottomrule
    \end{tabularx}
 	\label{ground_plot}
}
\end{table*}

\vspace{-1cm}
\subsection{Community Quality via Ground-truth}
We have ground-truth communities for the DBLP, Amazon, LiveJournal2, and Orkut networks, thus, for these networks, we compare against the ground-truth communities. Given a set of algorithmic communities $\mathcal{C}$ and the ground-truth communities $\mathcal{S}$, we compute $F_1$ measure and $F_2$ measure to evaluate the relevance between the algorithmic communities and the ground-truth communities. In general, $F_\beta$ measure is defined as follows:
\begin{displaymath}
F_\beta (\mathcal{S}_i) = (1+\beta^2) \dfrac{\text{\emph{precision}}(\mathcal{S}_i) \cdot \text{\emph{recall}}(\mathcal{S}_i)}{\beta^2 \cdot \text{\emph{precision}}(\mathcal{S}_i) + \text{\emph{recall}}(\mathcal{S}_i)}
\end{displaymath}
where $\beta$ is a non-negative real value, and the $\text{\emph{precision}}$ and $\text{\emph{recall}}$ of $\mathcal{S}_i \in \mathcal{S}$ are defined as follows:
\begin{displaymath}
\text{\emph{precision}}(\mathcal{S}_i) = \dfrac{|\mathcal{C}_j \bigcap \mathcal{S}_i|}{|\mathcal{C}_j|},
\end{displaymath}
\begin{displaymath}
\text{\emph{recall}}(\mathcal{S}_i) = \dfrac{|\mathcal{C}_j \bigcap \mathcal{S}_i|}{|\mathcal{S}_i|},
\end{displaymath}
where $\mathcal{C}_j \in \mathcal{C}$, and $F_\beta (\mathcal{S}_i) = F_\beta (\mathcal{S}_i, \mathcal{C}_{j^*})$ where $j^* = \underset{j}{\operatorname{argmax}}$ $F_\beta (\mathcal{S}_i, \mathcal{C}_j)$.
Then, the average $F_\beta$ measure is defined to be
\begin{displaymath}
\bar{F_\beta} = \dfrac{1}{|\mathcal{S}|} \sum_{\mathcal{S}_i \in \mathcal{S}} F_\beta(\mathcal{S}_i).
\end{displaymath}

Given an algorithmic community, $\text{\emph{precision}}$ indicates how many vertices are actually in the same ground-truth community. Given a ground-truth community, $\text{\emph{recall}}$ indicates how many vertices are predicted to be in the same community in a retrieved community. By definition, the precision and the recall are evenly weighted in $F_1$ measure. On the other hand, the $F_2$ measure puts more emphasis on recall than precision. The authors in \cite{yang-2013-bigclam} who provided the datasets argue that it is important to quantify the recall since the ground-truth communities in these datasets are partially annotated, i.e., some vertices are not annotated to be a part of the ground-truth community even though they actually belong to that community. This indicates that it would be reasonable to weight recall higher than precision, which is done by the $F_2$ measure.

In Table~\ref{ground_plot}, we report the average $F_1$ and $F_2$ measures on DBLP, Amazon, LiveJournal2, and Orkut networks. A higher value indicates better communities. We see that \NISE outperforms Bigclam, Demon, and Oslom in terms of both $F_1$ and $F_2$ measures on these networks. Within \NISE, ``spread hubs'' seeding is better than ``graclus centers'' seeding, and the standard PPR is slightly better than the Fiedler PPR in most of the cases. So, we see that the standard PPR is useful for identifying ground-truth communities. This result is also consistent with the recent observations in \cite{kloumann-2014-membership}.

\begin{table*}[t] 
{\scriptsize
\caption{Running Times of different methods on our test networks}
\label{runtime_table}
\centering
\begin{tabularx}{\textwidth}{XXXXXX}
\toprule
Graph & oslom & demon & bigclam & nise-sph-fppr & nise-grc-fppr \\ 
\midrule
{HepPh} & 19 mins. 16 secs. & 27 secs. & 11 mins. 23 secs. & 22 secs. & 2 mins. 48 secs. \\ 
{AstroPh} & 38 mins. 3 secs. & 42 secs. & 48 mins. 1 secs. & 36 secs. & 2 mins. 26 secs. \\ 
{CondMat} & 20 mins. 39 secs. & 50 secs. & 7 mins. 21 secs. & 36 secs. & 1 min. 14 secs. \\
{DBLP} & 5 hrs. 50 mins. & 3 hrs. 53 mins. & 7 hrs. 13 mins. & 18 mins. 20 secs. & 29 mins. 44 secs. \\ 
{Amazon} & 2 hrs. 55 mins. & 1 hr. 55 mins. & 1 hr. 25 mins. & 37 mins. 36 secs. & 42 mins. 43 secs. \\ 
{Flickr} & N/A & N/A & 69 hrs. 59 mins. & 43 mins. 55 secs. & 3 hrs. 56 mins. \\ 
{Orkut} & N/A & N/A & 13 hrs. 48 mins. & 1 hrs. 16 mins. & 4 hrs. 16 mins. \\
{LiveJournal} & N/A & N/A & 65 hrs. 30 mins. & 2 hrs. 36 mins. & 4 hrs. 48 mins. \\ 
{LiveJournal2} & N/A & N/A & 21 hrs. 35 mins. & 2 hrs. 15 mins. & 6 hrs. 37 mins. \\ 
{Myspace} & N/A & N/A & $>$ 7 days & 5 hrs. 27 mins. & 9 hrs. 42 mins. \\ \bottomrule  
\end{tabularx}
\label{runtime_table}
}
\end{table*}

\subsection{Comparison of Running Times}
Finally, we compare the running times of the different algorithms in Table~\ref{runtime_table}. To do a fair comparison, we run the single thread version of Bigclam and \NISE for HepPh, AstroPh, CondMat, DBLP, and Amazon networks. Since Demon and Oslom fail on larger networks, we use the multi-threaded version of Bigclam and \NISE with 4 threads for larger networks. We see that \NISE is the only method which can process the largest dataset (Myspace) in a reasonable time. On small networks (HepPh, AstroPh, and CondMat), ``nise-sph-fppr'' is faster than Demon, Oslom and Bigclam. On medium size networks (DBLP and Amazon), both ``nise-grc-fppr'' and ``nise-sph-fppr'' are faster than other methods. On large networks (Flickr, Orkut, LiveJournal, LiveJournal2, Myspace), \NISE is much faster than Bigclam.

\section{Discussion and Conclusion}
We now discuss the results from our experimental investigations. First, we note that \NISE is the only method that worked on all of the problems. Also, our method is faster than other state-of-the-art overlapping community detection methods. Perhaps surprisingly, the major difference in cost between using ``graclus centers'' for the seeds and the other seed choices does not result from the expense of running Graclus. Rather, it arises because the personalized PageRank expansion technique takes longer for the seeds chosen by Graclus. When the PageRank expansion methods has a larger input set, it tends to take longer, and the ``graclus centers'' seeding strategy is likely to produce larger input sets because of the neighborhood inflation and because the central vertices of clusters are likely to be high degree vertices.

We wish to address the relationship between our results and some prior observations on overlapping communities. The authors of Bigclam found that the dense regions of a graph reflect areas of overlap between overlapping communities. By using a conductance measure, we ought to find only these dense regions -- however, our method produces much larger communities that cover the entire graph. The reason for this difference is that we use the entire vertex neighborhood as the restart for the personalized PageRank expansion routine. We avoid seeding exclusively inside a dense region by using an entire vertex neighborhood as a seed, which grows the set beyond the dense region. Thus, the communities we find likely capture a combination of communities given by the ego network of the original seed node. 

Overall, \NISE significantly outperforms other state-of-the-art overlapping community detection methods in terms of runtime, conductance-vs-coverage, and ground-truth accuracy. Also, our new seeding strategies, ``graclus centers'' and ``spread hubs'', are superior than existing methods, thus play an important role in the success of our seed set expansion method.

\section*{Acknowledgments}
This research was supported by NSF grants CCF-1117055 and CCF-1320746 to ID, and by NSF CAREER award CCF-1149756 to DG.

\bibliographystyle{plain}
\bibliography{ref_nise}

\end{document}